\documentclass[
submission
,nomarks]{dmtcs-episciences}

\usepackage[utf8]{inputenc}
\usepackage{subfigure}
\usepackage[english]{babel}
\usepackage[T1]{fontenc}
\usepackage{amsmath,amssymb}
\usepackage[arrow, matrix, curve]{xy}
\usepackage{amsthm}
\usepackage{todonotes}
\usepackage{graphicx}
\usepackage{hyperref}
\usepackage{dsfont}
\usepackage{braket}
\usepackage{cite}
\usepackage[shortlabels]{enumitem}
\usepackage{tikz}
\usepackage{verbatim}
\usepackage{todonotes}

\newcommand{\pt}{bi\-colored~$P_3$}
\newcommand{\ptt}{bi\-colored-$P_3$}
\newcommand{\Pt}{Bicolored~$P_3$}
\newcommand{\BPDs}{\textsc{BPD}}

\newtheorem{theorem}{Theorem}[section]
\newtheorem{lemma}[theorem]{Lemma}
\newtheorem{proposition}[theorem]{Proposition}
\newtheorem{corollary}[theorem]{Corollary}

\newtheorem{claim}{Claim}
\theoremstyle{definition}
\newtheorem{defi}[theorem]{Definition}

\newtheorem{reduc}{Reduction Rule}{\bfseries}{\itshape}
\newtheorem{branch}{Branching Rule}{\bfseries}{\itshape}
\theoremstyle{definition}
\newtheorem*{claimproof}{\normalfont{\textit{Proof}}}

\definecolor{myred}{rgb}{1,0.25,0.25}

\newcommand{\Oh}{\ensuremath{\mathcal{O}}}

\newcommand{\cC}{\mathcal{C}}

\newcommand{\proofparagraph}[1]{\par\smallskip\textit{#1}}

\author{Niels Gr\"uttemeier
	\and Christian Komusiewicz \\
	\and Jannik Schestag
	\and Frank~Sommer\thanks{FS was supported by the DFG, project MAGZ (KO~3669/4-1).}}
	
\title{Destroying \Pt s by Deleting Few~Edges}
  
\affiliation{Fachbereich Mathematik und Informatik, Philipps-Universität Marburg, Marburg, Germany}


\received{2020-02-17}

\revised{2021-04-09}

\accepted{2021-05-07}

\begin{document}
\publicationdetails{23}{2021}{1}{14}{6108}
\maketitle
\begin{abstract} We introduce and study the \textsc{\Pt{} Deletion} problem defined as
  follows. The input is a graph~$G=(V,E)$ where the edge set~$E$ is partitioned into a set~$E_r$ of red
  edges and a set~$E_b$ of blue edges. The question is whether we can delete at
  most~$k$ edges such that~$G$ does not contain a bicolored~$P_3$ as an induced
  subgraph. Here, a bicolored~$P_3$ is a path on three vertices with one blue and one red
  edge. We show that \textsc{\Pt{} Deletion} is NP-hard and cannot be solved
  in~$2^{o(|V|+|E|)}$~time on bounded-degree graphs if the ETH is true. Then, we
  show that \textsc{\Pt{} Deletion} is polynomial-time solvable when~$G$ does not contain a
  bicolored~$K_3$, that is, a triangle with edges of both colors. We also provide a polynomial-time algorithm for the case that~$G$ contains no blue~$P_3$, red~$P_3$, blue~$K_3$, and
  red~$K_3$. Finally, we show that~\textsc{\Pt{} Deletion} can be solved
  in~$\Oh(1.84^k\cdot |V| \cdot |E|)$~time and that it admits a kernel with~$\Oh(k\Delta\min(k,\Delta))$
  vertices, where~$\Delta$ is the maximum degree of~$G$. 
  \end{abstract}
  
\subsection*{Acknowledgment}

We would like to thank the reviewers of \emph{Discrete Mathematics and Theoretical Computer Science} for their helpful comments and  Michał Pilipczuk (University of Warsaw) for
 pointing out the connection to Gallai colorings and their characterization. A preliminary version of this work appeared in \emph{Proceedings of the 15th Conference on Computability
               in Europe (CiE~'19)}, volume 11558 of \emph{Lecture Notes in Computer Science}, pages 193--204. The full version contains all missing proofs and an improved running time analysis of the fixed-parameter algorithm. 
  Some of the results of this work are also contained in the third author's Bachelor thesis~\cite{Sche19}. 

\section{Introduction}
Graph modification problems are a popular topic in computer science. In these problems,
one is given a graph and wants to apply a minimum number of modifications, for example
edge deletions, to obtain a graph that fulfills some graph property~$\Pi$.

An important reason for the popularity of graph modification problems
is their usefulness in graph-based data analysis. A classic problem in this context is 
\textsc{Cluster Editing} where we may insert and delete edges and~$\Pi$ is the set of
cluster graphs. These are exactly the graphs that are disjoint unions of cliques and it is
well-known that a graph is a cluster graph if and only if it does not contain a~$P_3$, a path on three
vertices, as induced subgraph. \textsc{Cluster Editing} has many applications~\cite{BB13}, for example
in clustering gene interaction networks~\cite{BSY99} or protein
sequences~\cite{WBLR07}. The variant where we may only delete edges is known as
\textsc{Cluster Deletion}~\cite{SST04}. Further graph-based data analysis problems that
lead to graph modification problems for some graph property~$\Pi$ defined by small
forbidden induced subgraphs arise in the analysis of biological~\cite{BHK15,HWL+15} or
social networks~\cite{BHSW15,NG13}.

  Besides the application, there is a more theoretical reason why graph modification problems are very important in computer science: Often these problems are
NP-hard~\cite{LY80,Yan81} and thus they represent interesting case studies for algorithmic
approaches to NP-hard problems.  For example, by systematically categorizing graph
properties based on their forbidden subgraphs one may outline the border between tractable
and hard graph modification problems~\cite{ASS17,Kom18,Yan81}.

In recent years, multilayer graphs have become an increasingly important tool for
integrating and analyzing network data from different sources~\cite{KAB+14}. Formally,
multilayer graphs can be viewed as edge-colored (multi-)graphs, where each edge color represents one
layer of the input graph. With the advent of multilayer graphs in network analysis it can be
expected that graph modification problems for edge-colored graphs will arise in many applications as it was the case in uncolored graphs.

One example for such a problem is \textsc{Module
  Map}~\cite{SK18}. Here, the input is a simple graph with red
and blue edges and the aim is to obtain by a minimum number of edge deletions and
insertions a graph that contains no~$P_3$ with two blue edges, no~$P_3$ with a red and a
blue edge, and no a triangle, called~$K_3$, with two blue edges and one red edge. \textsc{Module Map} arises in computational biology~\cite{AS14,SK18}; the red layer represents genetic interactions
and the blue layer represents physical protein interactions~\cite{AS14}. 

Motivated by the practical application of \textsc{Module Map}, an edge deletion problem
with bicolored forbidden induced subgraphs, we aim to study such problems from a more
systematic and algorithmic point of view. Given the importance of~$P_3$-free graphs in the
uncolored case, we focus on the problem where we want to destroy all
\emph{bicolored~$P_3$s}, that is, all~$P_3$s with one blue and one red edge, by edge
deletions.
\begin{quote}
  \textsc{\Pt{} Deletion (BPD)}\\
  \textbf{Input}: A two-colored graph~$G=(V, E_r,E_b)$ 
  and an integer~$k \in \mathds{N}$.\\
  \textbf{Question}: Can we delete at most~$k$ edges from~$G$ such that the remaining graph contains no \pt{} as induced subgraph?
\end{quote}

We use~$E:=E_r\uplus E_b$ to denote the set of all edges of~$G$,~$n:=|V|$ to denote the number of vertices in~$G$, and~$m:=|E|$ to denote the number of edges in~$G$.

Bicolored $P_3$s are closely connected to Gallai colorings of complete graphs~\cite{Gal67,GS04}. A Gallai coloring is an edge-coloring such that the edges of every triangle receive at most two different colors. When we view nonedges of~$G$ as edges with a third color, say green, then a bicolored~$P_3$ is the same as a triangle that violates the property of Gallai colorings. Thus, \BPDs{} is essentially equivalent to the following problem: Given a complete graph with an edge-coloring with the colors red, blue, and green that is not a Gallai coloring, can we transform the coloring into a Gallai coloring by recoloring at most~$k$ blue or red edges with the color green?
\paragraph{Our Results.}  We show that \BPDs{} is NP-hard and that, assuming the Expo\-nential-Time Hypo\-thesis (ETH)~\cite{IPZ01}, it cannot be solved
in a running time that is subexponential in the instance size. We then study two different
aspects of the computational complexity of the problem. 

First, we consider special cases that can be solved in polynomial time, motivated by
similar studies for problems on uncolored graphs~\cite{BLS99}. We are in particular
interested in whether or not we can exploit structural properties of input graphs that can
be expressed in terms of \emph{colored} forbidden subgraphs. We show that \BPDs{} can be solved
in polynomial time on graphs that do not contain a certain type of bicolored~$K_3$s as induced subgraphs, where bicolored~$K_3$s are triangles with edges of both colors. Moreover, we show that \BPDs{} can be solved in polynomial time on graphs that contain no~$K_3$s with one edge color and no~$P_3$s
with one edge color as induced subgraphs.

Second, we consider the parameterized complexity of \BPDs{} with respect to the natural
parameter~$k$. We show that \BPDs{} can be solved in~$\Oh(1.84^k\cdot nm )$~time and
that it admits a problem kernel with $\Oh(k\Delta\min(k,\Delta))$~vertices, where~$\Delta$ is the
maximum degree in~$G$. As a side result, we show that \BPDs{} admits a trivial problem
kernel with respect to~$\ell:= m-k$.



\section{Preliminaries}
We consider undirected simple graphs~$G$ with vertex set~$V$ and edge set~$E$, where~$E$ is partitioned into a set~$E_b$ of \emph{blue edges} and a set~$E_r$ of \emph{red edges}, denoted by~$G=(V,E_r,E_b)$. For a vertex~$v$,~$N_G(v):=\{u\mid \{u,v\}\in E\}$ denotes the \emph{open
  neighborhood} of~$v$ and~$N_G[v]:=N_G(v) \cup \{v\}$ denotes the \emph{closed neighborhood} of~$v$. For a vertex set~$W$,~$N_G(W):=\bigcup_{w\in W}N(w)\setminus W$ denotes the \emph{open neighborhood} of~$W$ and~$N_G[W]:=N_G(W)\cup W$ denotes the \emph{closed neighborhood} of~$W$. The \emph{degree}~$\deg(v):=|N_G(v)|$ of a vertex~$v$ is the size of its open neighborhood. We let~$N_G^2(v):=N_G(N_G(v))\setminus \{v\}$ denote the \emph{second neighborhood} of~$v$. For any two vertex sets~$V_1, V_2 \subseteq V$, we denote by~$E_G(V_1,V_2):= \{ \{v_1,v_2\} \in E \mid v_1 \in V_1, v_2 \in V_2 \}$ the set of edges between~$V_1$ and~$V_2$ in~$G$ and write~$E_G(V'):=E_G(V',V')$. In each context we may omit the subscript~$G$ if the graph is clear from the context.

For any~$V'\subseteq V$,~$G[V']:=(V',E(V')\cap E_r, E(V') \cap E_b)$ denotes the \emph{subgraph induced by~$V'$}. We say that some graph~$H=(V^H,E^H_r,E_b^H)$ is an \emph{induced subgraph} of $G$ if there is a set~$V' \subseteq V$, such that~$G[V']$ is isomorphic to~$H$, otherwise~$G$ is called \emph{$H$-free}. Two vertices~$u$ and~$v$ are \emph{connected} if there is a path from~$u$ to~$v$ in~$G$. A \emph{connected component} is a maximal vertex set~$S$ such that each two vertices are connected in~$G[S]$. A \emph{clique} in a graph~$G$ is a set~$K \subseteq V$  of vertices such that in~$G[K]$ each pair of vertices is adjacent. The graph~$(\{u,v,w\},\{\{u,v\}\},\{\{v,w\}\})$ is called \emph{\pt}. We say that a vertex~$v\in V$ is \emph{part of} a \pt~in~$G$ if there is a set~$V' \subseteq V$ with~$v \in V'$ such that~$G[V']$ is a \pt. Furthermore, we say that two edges~$\{u,v\}$ and~$\{v,w\}$ \emph{form} a \pt~if~$G[\{u,v,w\}]$ is a \pt. An edge~$e$ is~\emph{part of a~\pt} if there exists some other edge~$e'$ such that~$e$ and~$e'$ form a~\pt. For any edge set~$E'$ we denote by~$G-E':=(V, E_r \setminus E', E_b \setminus E')$ the graph we obtain by deleting all edges in~$E'$. As a shorthand, we write~$G-e:=G-\{e\}$ for an edge~$e$. An edge deletion set~$S$ is a \emph{solution} for an instance~$(G,k)$ of \BPDs~if~$G-S$ is \ptt-free and~$|S|\le k$.

A \emph{branching rule} for some problem~$L$ is a computable function that maps an instance~$w$ of~$L$ to a tuple of instances~$(w_1, \dots, w_t)$ of~$L$. A branching rule is called \emph{correct} if~$w$ is a yes-instance for~$L$ if and only if there is some~$i \in \{1,\dots,t\}$ such that~$w_i$ is a yes-instance of~$L$. The application of branching rules gives rise to a search tree whose size is analyzed using branching vectors; for more details refer to the textbook of Fomin and Kratsch~\cite{FK10}.  A \emph{reduction rule} for some problem~$L$ is a computable function that maps an instance~$w$ of~$L$ to an instance~$w'$ of~$L$ such that~$w$ is a yes-instance if and only if~$w'$ is a yes-instance. 

\emph{Parameterized Complexity} is the analysis of the complexity of problems depending on the input size~$n$ and a problem parameter~$k$~\cite{Cyg+15,DF13}. A problem is called \emph{fixed-parameter tractable} if there exists an algorithm with running time~$f(k) \cdot n^{\Oh(1)}$ for some computable function~$f$ that solves the problem. An important tool in the development of parameterized algorithms is \emph{problem kernelization}. Problem kernelization is a polynomial-time preprocessing by \emph{reduction rules}: A problem~$L$ admits a problem kernel if, given any instance~$I$ of~$L$ with parameter~$k$, one can compute an equivalent instance~$I'$ of~$L$ with parameter~$k'$ in polynomial time such that~$k' \leq k$ and the size of~$I'$ is bounded by some computable function~$g$ only depending on~$k$. The function~$g$ is called~\emph{kernel size}. The \emph{Exponential Time Hypothesis (ETH)} is a standard complexity theoretical conjecture used to prove lower bounds.
It implies that~\textsc{3-SAT} cannot be solved in~$2^{o(|\phi|)}$~time where~$\phi$ denotes the input formula~\cite{IPZ01}.

\section{Bicolored $\mathbf{P_3}$ Deletion is NP-hard}
In this section we prove the NP-hardness of \BPDs. This motivates our study of polynomial-time solvable cases and the parameterized complexity in Sections~\ref{Section: polytime} and~\ref{Section: param}, respectively.
\begin{theorem} \label{Theorem: NP-h}
\BPDs{} is NP-hard even if the maximum degree of~$G$ is~$8$.
\end{theorem}

\begin{proof}
We present a polynomial-time reduction from the NP-hard \textsc{(3,4)-SAT} problem where one is given a 3-CNF formula~$\phi$ where each variable occurs in at most four clauses, and the question is if there is a satisfying assignment for~$\phi$~\cite{Tovey84}.

Let~$\phi$ be a 3-CNF formula with variables~$X=\{x_1, \dots, x_{|X|}\}$ and clauses~$\cC=\{C_1, \dots, C_{|\cC|}\}$ with four occurrences per variable. For a given variable $x_i$ that occurs in a clause $C_j$ we define the \emph{occurrence number} $\Psi (C_j,x_i)$ as the number of clauses in $\{C_1, C_2, \dots, C_j \}$ where~$x_i$ occurs. Intuitively, $\Psi(C_j,x_i)=r$ means that the $r$th occurrence of variable~$x_i$ is the occurrence in clause~$C_j$. Since each variable occurs in at most four clauses, we have~$\Psi(C_j,x_i) \in \{1,2,3,4\}$.

\proofparagraph{Construction:} We describe how to construct an equivalent instance~$(G=(V,E_r,E_b),k)$ of \BPDs{} from~$\phi$.

For each variable~$x_i \in X$ we define a \emph{variable gadget} as follows. The variable gadget of~$x_i$ consists of a \emph{central vertex}~$v_i$ and two vertex sets~$T_i:=\{t_i^1,t_i^2,t_i^3,t_i^4\}$ and~$F_i:=\{f_i^1,f_i^2,f_i^3,f_i^4\}$. We add a blue edge from~$v_i$ to every vertex in~$T_i$ and a red edge from~$v_i$ to every vertex in~$F_i$.

For each clause~$C_j \in \mathcal{C}$ we define a \emph{clause gadget} as follows. The clause gadget of~$C_j$ consists of three vertex sets~$A_j:=\{a_j^1,a_j^2,a_j^3\}$, $B_j:=\{b_j^1,b_j^2,b_j^3\}$, and~$W_j:=\{w_j^1, w_j^2, w_j^3 ,w_j^4\}$. We add blue edges such that the vertices in~$B_j\cup W_j$ form a clique with only blue edges in~$G$. Moreover, for each~$p \in \{1,2,3\}$, we add a blue edge~$\{a_j^p,b_j^p\}$ and a red edge~$\{a_j^p,u\}$ for every~$u \in W_j \cup B_j \setminus \{b_j^p\}$. Observe that there are no edges between $a_j^1$, $a_j^2$, and $a_j^3$; all other vertex pairs are connected either by a red edge or a blue edge.

We connect the variable gadgets with the clause gadgets by identifying vertices in~$T_i \cup F_i$ with vertices in~$A_j$ as follows. Let~$C_j$ be a clause containing variables~$x_{i_1}, x_{i_2}$, and~$x_{i_3}$. For each~$p \in \{1,2,3\}$ we set~
$$a_j^p = \begin{cases}
t_{i_p}^{\Psi(C_j,x_{i_p})}&\text{if }x_{i_p}\text{ occurs as a positive literal in }C_j\text{, and}\\
f_{i_p}^{\Psi(C_j,x_{i_p})}&\text{if }x_{i_p}\text{ occurs as a negative literal in }C_j\text{.}
\end{cases}$$

Now, for every variable~$x_i \in X$ each vertex in~$T_i \cup F_i$ is identified with at most one vertex~$a_j^p$. Figure~\ref{Figure: NP-h construction} shows an example of a clause gadget and its connection with the variable gadgets. To complete the construction of the BPD instance~$(G,k)$ we set~$k:=4\cdot |X|+14 \cdot |\cC|$.

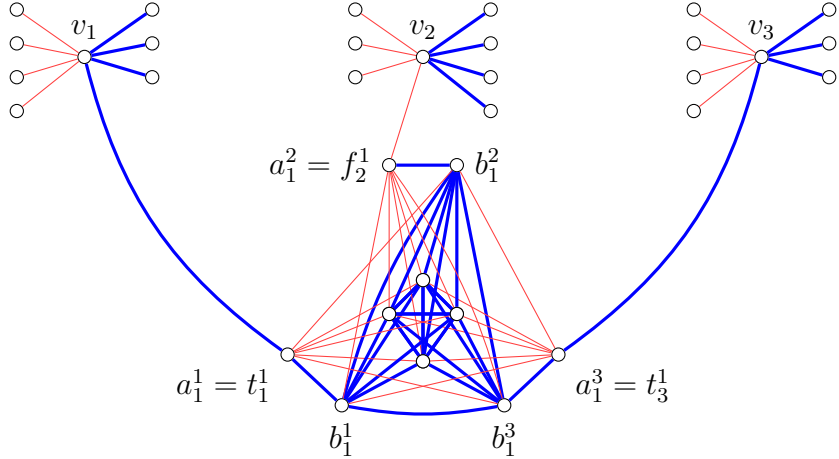
\begin{figure}
\begin{center}
\begin{tikzpicture}[scale=0.9]
\tikzstyle{knoten}=[circle,fill=white,draw=black,minimum size=5pt,inner sep=0pt]
\tikzstyle{bez}=[inner sep=0pt]
		\node[knoten] (a1j)[label=210:{$a^1_1=t_1^1$}] at (-2,-0.4) {};
                \node[knoten] (a2j)[label=left:{$a^2_1=f_2^1$}] at (-0.5,2.4) {};
		\node[knoten] (a3j)[label=330:{$a^3_1=t_3^1$}] at (2,-0.4) {};

                \node[knoten] (b1j)[label=270:{$b^1_1$}] at (-1.2,-1.15) {};
		\node[knoten] (b2j)[label=right:{$b^2_1$}] at (0.5,2.4) {};
		\node[knoten] (b3j)[label=270:{$b^3_1$}] at (1.2,-1.15) {};

		

		\node[knoten] (w1) at (0,0.7) {};
		\node[knoten] (w2) at (0,-0.5) {};
		\node[knoten] (w3) at (-0.5,0.2) {};
		\node[knoten] (w4) at (0.5,0.2) {};
		
		\draw[-,color=blue,very thick] (a1j) to (b1j);
		\draw[-,color=blue,very thick] (a2j) to (b2j);
		\draw[-,color=blue,very thick] (a3j) to (b3j);

		\draw[-,color=blue,very thick,bend left=10] (b1j) to (b2j);
		\draw[-,color=blue,very thick,bend right=10] (b1j) to (b3j);
                \draw[-,color=blue,very thick] (b2j) to (b3j);
		
		\foreach \x in {1,2,...,4}{
         \draw[-, color=myred] (a1j) to (w\x);
         \draw[-, color=blue, very thick] (b1j) to (w\x);
         \draw[-, color=myred] (a2j) to (w\x);
         \draw[-, color=blue, very thick] (b2j) to (w\x);
         \draw[-, color=myred] (a3j) to (w\x);
         \draw[-, color=blue, very thick] (b3j) to (w\x);
         \foreach \y in {1,2,...,4}{
         \draw[-, color=blue, very thick] (w\y) to (w\x);
         }
      	}
      	\draw[-, color=myred] (a1j) to (b2j);
      	\draw[-, color=myred] (a1j) to (b3j);
      	\draw[-, color=myred] (a2j) to (b1j);
      	\draw[-, color=myred,bend left=10] (a2j) to (b3j);
      	\draw[-, color=myred] (a3j) to (b1j);
      	\draw[-, color=myred] (a3j) to (b2j);
      


		\node[knoten] (v1)[label=above:{$v_1$}] at (-5.0,4.0) {};
		\node[knoten] (v2)[label=above:{$v_2$}] at (00,4.0) {};
		\node[knoten] (v3)[label=above:{$v_3$}] at (5.0,4.0) {};
		
		\node[knoten] (t11) at (-4.0,4.7) {};
		\node[knoten] (t12) at (-4.0,4.2) {};
		\node[knoten] (t13) at (-4.0,3.7) {};
		\node[knoten] (f11) at (-6.0,4.7) {};
		\node[knoten] (f12) at (-6.0,4.2) {};
		\node[knoten] (f13) at (-6.0,3.7) {};
		\node[knoten] (f14) at (-6.0,3.2) {};
		
		\node[knoten] (f21) at (-1.0,4.7) {};
		\node[knoten] (f22) at (-1.0,4.2) {};
		\node[knoten] (f23) at (-1.0,3.7) {};
		\node[knoten] (t24) at (1.0,3.2) {};
		\node[knoten] (t21) at (1.0,4.7) {};
		\node[knoten] (t22) at (1.0,4.2) {};
		\node[knoten] (t23) at (1.0,3.7) {};
		
		\node[knoten] (t31) at (6.0,4.7) {};
		\node[knoten] (t32) at (6.0,4.2) {};
		\node[knoten] (t33) at (6.0,3.7) {};
		\node[knoten] (f31) at (4.0,4.7) {};
		\node[knoten] (f32) at (4.0,4.2) {};
		\node[knoten] (f33) at (4.0,3.7) {};
		\node[knoten] (f34) at (4.0,3.2) {};
		
		\foreach \x in {1,2,3}{
			\foreach \y in {1,2,3}{
         		\draw[-, color=myred] (v\x) to (f\x\y);
         		\draw[-, color=blue,very thick] (v\x) to (t\x\y);
         	}
        }
        \draw[-, color=myred] (v1) to (f14);
        \draw[-, color=blue, very thick, bend right=20] (v1) to (a1j);
        
        \draw[-, color=blue,very thick] (v2) to (t24);
        \draw[-, color=myred] (v2) to (a2j);
        
		\draw[-, color=myred] (v3) to (f34);
        \draw[-, color=blue, very thick, bend left=20] (v3) to (a3j);
		
		\node[knoten] (w1) at (0,0.7) {};
		\node[knoten] (w2) at (0,-0.5) {};
		\node[knoten] (w3) at (-0.5,0.2) {};
		\node[knoten] (w4) at (0.5,0.2) {};
		
\end{tikzpicture}
\end{center} 
\caption{The lower part of the figure shows the clause gadget of a clause~$C_1=(x_1 \lor \overline{x_2} \lor x_3)$. The upper part of the figure shows variable gadgets representing variables~$x_1$, $x_2$, and~$x_3$. The vertices~$a^1_1, a^2_1$, and~$a^3_1$ from the clause gadget are identified with vertices from the variable gadgets. The bold lines represent blue edges and the thin lines represent red edges.} \label{Figure: NP-h construction}
\end{figure}

\proofparagraph{Intuition:}
Before showing the correctness of the reduction, we describe its idea. For each variable~$x_i$ we have to delete all blue edges in~$E(\{v_i\},T_i)$ or all red edges in~$E(\{v_i\},F_i)$ in the corresponding variable gadget. Deleting the edges in~$E(\{v_i\},T_i)$ assigns true to the variable~$x_i$ while deleting the edges in~$E(\{v_i\},F_i)$ assigns false to~$x_i$. Since we identify vertices in~$T_i \cup F_i$ with vertices in~$A_j$ the information of the truth assignment is transmitted to the clause gadgets. We will be able to make a clause-gadget \ptt-free with~$14$ edge deletions if and only if there is at least one vertex in~$A_j$ which is incident with a deleted edge of its variable gadget.

\proofparagraph{Correctness:} We now show the correctness of the reduction by proving that there is a satisfying assignment for~$\phi$ if and only if~$(G,k)$ is a yes-instance of \BPDs.

$(\Rightarrow)$ Let~$\mathcal{A}:X \rightarrow \{\text{true},\text{false}\}$ be a satisfying assignment for~$\phi$. In the following, we construct a solution~$S$ for~$(G,k)$.

For each variable~$x_i$, we add~$E(\{v_i\},T_i)$ to~$S$ if~$\mathcal{A}(x_i)=\text{true}$ and we add~$E(\{v_i\},F_i)$ to~$S$ if~$\mathcal{A}(x_i)=\text{false}$. Note that for each variable we add exactly four edges to~$S$.
For each $C_j \in \cC$ we add the following edges: Since~$\mathcal{A}$ is satisfying, $C_j$ contains a variable~$x_i$ such that~$\mathcal{A}(x_i)$ satisfies~$C_j$. By the construction of~$G$ there is exactly one~$p \in \{1,2,3\}$ such that~$a_j^p=t^{\Psi(C_j,x_i)}_i$ if~$x_i$ occurs as a positive literal in~$C_j$ or~$a_j^p=f^{\Psi(C_j,x_i)}_i$ if~$x_i$ occurs as a negative literal in~$C_j$. For both~$q \in \{1,2,3\} \setminus \{p\}$ we add~$E(\{a^{q}_j\},W_j\cup B_j)$ to~$S$. Note that since~$|E(\{a^{q}_j\},W_j\cup B_j)| = |W_j\cup B_j| =7$ we add exactly~$14$ edges per clause. Thus, we have an overall number of~$4 \cdot |X| +14\cdot |\cC|$ edges in~$S$.

Let~$G':=G-S$. It remains to show that there is no \pt~in~$G'$. Let~$C_j \in \cC$. We first show that no edge in~$E_{G'}(A_j \cup B_j \cup W_j)$ is part of a \pt~in~$G'$. For any two vertices~$u_1,u_2 \in W_j\cup B_j$ it holds that~$N_{G'}[u_1]=N_{G'}[u_2]$. Hence, no edge in~$E_{G'}(W_j \cup B_j)$ is part of an induced~$P_3$ in~$G'$ and therefore no edge in~$E_{G'}(W_j \cup B_j)$ is part of a \pt~in~$G'$. It remains to show that no edge in~$E_{G'}(A_j,W_j \cup B_j)$ is part of a \pt~in~$G'$. Without loss of generality assume that~$E(\{a^{1}_j\} \cup \{a^{2}_j\} ,W_j\cup B_j) \subseteq S$ and $E(\{a^{3}_j\} ,W_j\cup B_j) \cap S = \emptyset$. Thus, the only possible edges in~$E_{G'}(A_j \cup B_j \cup W_j)$ that might form a \pt~in~$G'$ are the edges in~$E(\{a^{3}_j\} ,W_j\cup B_j)$. We show for every~$u \in W_j \cup B_j$ that~$\{a^3_j,u\} \in E_{G'}(A_j,W_j \cup B_j)$ is not part of an induced \pt~in~$G'$. Since~$\{a_j^3\} \cup B_j \cup W_j$ is a clique in~$G'$, the edge $\{a^3_j,u\}$ may only form a \pt~with an edge~$\{a_j^3,v_i\}$, where~$v_i$ is the central vertex of the variable gadget of a variable~$x_i$ occurring in $C_j$. By the construction of~$G$ it follows that $a^3_j \in \{ t^{\Psi(C_j,x_i)}_i, f^{\Psi(C_j,x_i)}_i \}$ for exactly one variable~$x_i$. Since~$E(\{a^{3}_j\} ,W_j\cup B_j) \cap S = \emptyset$, the assignment~$\mathcal{A}(x_i)$ satisfies clause~$C_j$ by the construction of~$S$. If~$\mathcal{A}(x_i)=\text{true}$, then~$\{a^3_j, v_i\} = \{t_i^{\Psi(C_j,x_i)}, v_i\}\in E_G(\{v_i\},T_i) \subseteq S$ and therefore~$\{a^3_j, v_i\}$ is not an edge of~$G'$. Analogously, if~$\mathcal{A}(x_i)=\text{false}$, then~$\{a^3_j, v_i\} = \{f_i^{\Psi(C_j,x_i)}, v_i\}\in E_G(\{v_i\},F_i) \subseteq S$ is not an edge of~$G'$. Hence, no edge in~$E_{G'}(A_j \cup B_j \cup W_j)$ is part of a \pt~in~$G'$.

Let~$x_i \in X$. We show that no edge in~$E_{G'}( \{v_i\} \cup T_i \cup F_i)$ is part of a \pt~in~$G'$. Since either~$E_{G}( \{v_i\}, T_i) \subseteq S$ or~$E_{G}( \{v_i\}, F_i) \subseteq S$, all edges in~$E_{G'}( \{v_i\} \cup T_i \cup F_i)$ have the same color. Hence, there is no \pt~in~$G'$ consisting of two edges from one variable gadget. Since there is no vertex in~$G'$ that is adjacent to two vertices of distinct variable gadgets, an edge~$e \in E_{G'}( \{v_i\} \cup T_i \cup F_i)$ may only form a \pt~with an edge in~$E_{G'}(A_j \cup B_j \cup W_j)$ for some clause~$C_j$. However, since no edge in~$E_{G'}(A_j \cup B_j \cup W_j)$ is part of a \pt~in~$G'$ as shown above, $e$ does not form a \pt~with an edge from a clause gadget. Therefore, no edge in $E_{G'}( \{v_i\} \cup T_i \cup F_i)$ is part of a \pt. It follows that~$G'$ does not contain any \pt.

$(\Leftarrow)$ Conversely, let~$S$ be a solution for~$(G,k)$. For every variable~$x_i\in X$ we have~$|S \cap E_G(\{v_i\},T_i \cup F_i)| \geq 4$ since every edge in~$E_G(\{v_i\},T_i)$ forms a \pt~with every edge in~$E_G(\{v_i\},F_i)$. 

Before we define a satisfying assignment~$\mathcal{A}: X \rightarrow \{\text{true}, \text{false}\}$ for~$\phi$, we take a more detailed look at the edges of the clause gadgets that need to be in~$S$. Let~$C_j \in \cC$ be a clause and let~$G_j:=G[A_j\cup B_j\cup W_j]$ be the induced subgraph of the corresponding clause gadget. We show that~$14$ edge deletions are necessary and sufficient to transform~$G_j$ into a \ptt-free graph. Obviously, for pairwise distinct~$p,q,r \in \{1,2,3\}$, deleting the~$14$ edges in~$E_{G_j}(\{a_j^p,a_j^q\},B_j\cup W_j)$ transforms~$G_j$ into a \ptt-free graph, since~$\{a_j^r\} \cup B_j \cup W_j$ is a clique in~$G_j$. Hence, deleting~$14$ edges is sufficient. It remains to show that when deleting less than~$14$ edges there are still \pt s in~$G_j$. To this end, we show that either one of the vertices in~$A_j$ is not incident with an edge deletion in~$G_j$ or we need more than~$14$ edge deletions to transform~$G_j$ into a \ptt-free graph. We consider three vertices~$u_1, u_2, u_3 \in B_j \cup W_j$ representing the endpoints of deleted edges incident with~$a_j^1$, $a_j^2$, and~$a_j^3$, respectively. Let~$S_j:=\{ \{a_j^p, u_p \} \mid p \in \{1,2,3\}\}$. The following claim gives a lower bound on the number of edge deletions in~$G_j$ after deleting~$S_j$.

\begin{claim} \label{Claim: One clause guy without edgedel}
There are at least~$12$ edge-disjoint~\pt s in~$G_j-S_j$.
\end{claim}
\begin{claimproof}
We define three sets~$\mathcal{P}^1$,~$\mathcal{P}^2$, and~$\mathcal{P}^3$ containing~\pt s and show that the union $\mathcal{P}^1 \cup \mathcal{P}^2 \cup \mathcal{P}^3$ contains at least~$12$ edge-disjoint \pt s in~$G_j - S_j$. Here, we represent \pt s by edge sets of size two. For each~$p \in \{1,2,3\}$ we set
\begin{align*}
\mathcal{P}^p := \{  \{ \{a_j^p, w\}, \{w,u_p\} \} \mid w \in (B_j \cup W_j) \setminus \{b_j^p, u_p\}  \}.
\end{align*}
Let~$p \in \{1,2,3\}$. Since $E_{G_j}(W_j \cup B_j) \subseteq E_b$, and~$E_{G_j}( \{a_j^p\}, (B_j \cup W_j) \setminus \{b_j^p\}) \subseteq E_r$, it follows that every set~$P \in \mathcal{P}^p$ is a \pt~in~$G_j- S_j$. Obviously, the \pt s in~$\mathcal{P}^p$ are edge-disjoint and $|\mathcal{P}^p| = |(B_j \cup W_j) \setminus \{b_j^1,u_1\}| = 5$. 

We now show that the union~$\mathcal{P}^1 \cup \mathcal{P}^2 \cup \mathcal{P}^3$ contains at least~$12$ edge-disjoint \pt s in~$G_j - S_j$. To this end, consider the following subset
\begin{align*}
\mathbb{P}:= (\mathcal{P}^1 \cup \mathcal{P}^2 \cup \mathcal{P}^3) \setminus \{Q_1,Q_2,Q_3\} \subseteq \mathcal{P}^1 \cup \mathcal{P}^2 \cup \mathcal{P}^3,
\end{align*}
where
\begin{align*}
Q_1 &:= \{ \{a_j^2, u_1\}, \{u_1, u_2\} \},\\
Q_2 &:= \{ \{a_j^3, u_1\}, \{u_1, u_3\} \},\\
Q_3 &:= \{ \{a_j^3, u_2\}, \{u_2, u_3\} \}.
\end{align*}
Obviously,~$Q_1, Q_2, Q_3 \not \in \mathbb{P}$ and $|\mathbb{P}| \geq 3 \cdot 5-3=12$. It remains to show that all \pt s in~$\mathbb{P}$ are edge-disjoint. Assume towards a contradiction that there are~$P,R \in \mathbb{P}$ with~$P \neq R$ and~$P \cap R \neq \emptyset$. Since every~$\mathcal{P}^p$ contains edge-disjoint \pt s, it follows that~$P \in \mathcal{P}^{p}$ and~$R \in \mathcal{P}^{r}$ for some~$p \neq r$. Without loss of generality assume~$p < r$. Since for every~$w \in B_j \cup W_j$, the edges~$\{a_j^{r}, w \}$ are not part of any \pt~in~$\mathcal{P}^{p}$ and, conversely, the edges~$\{a_j^{p}, w \}$ are not part of any \pt~in~$\mathcal{P}^{r}$ it follows that~$P \cap R = \{\{u_{p},u_{r}\}\}$. We conclude~$R=\{  \{a^{r}_j, u_{p}\}, \{u_{p}, u_{r}\}\}$.

Consider the case~$p = 1$ and~$r =2$. Then,~$R=Q_1 \not \in \mathbb{P}$. Analogously, if~$p = 1$ and~$r =3$, then~$R=Q_2 \not \in \mathbb{P}$, and if~$p=2$ and~$r=3$, then~$R=Q_3 \not \in \mathbb{P}$. In every case we have~$R \not \in \mathbb{P}$ which contradicts the assumption~$P, R \in \mathbb{P}$. Hence, there are no \pt s in~$\mathbb{P}$ that share an edge and therefore~$\mathcal{P}^1 \cup \mathcal{P}^2 \cup \mathcal{P}^3$ contains at least~$12$ edge-disjoint \pt s as claimed. $\hfill \Diamond$ 

\end{claimproof}

Claim \ref{Claim: One clause guy without edgedel} implies that if every vertex in~$A_j$ is incident with an edge in~$S \cap E_{G_j}(A_j,B_j \cup W_j)$, we have $|S \cap E_{G_j}(A_j \cup B_j \cup W_j)| \geq 3+12=15$.
We now show that deleting~$E_{G_j}(\{a_j^p,a_j^q\},B_j\cup W_j)$ for distinct~$p,q \in \{1,2,3\}$ are the only three possible ways to transform~$G_j$ into a \ptt-free graph with less than~$15$ edge deletions. By Claim~\ref{Claim: One clause guy without edgedel}, we can assume without loss of generality that~$E_{G_j}(\{a_j^3\},W_j \cup B_j) \cap S = \emptyset$.  We show that this implies that all edges incident with~$a^1_j$ and~$a^2_j$ in~$G_j$ are deleted by~$S$.

\begin{claim} \label{Claim: NP-h all edges deleted}
If~$E_{G_j}(\{a_j^3\},W_j \cup B_j) \cap S = \emptyset$, then~$E_{G_j}(\{a_j^p\},B_j \cup W_j) \subseteq S$ for~$p=1$ and for~$p=2$.
\end{claim}

\begin{claimproof}
First, note that no edge~$\{b^3_j, w\}$ with $w \in B_j \cup W_j$ is an element of~$S$, since otherwise~$\{a^3_j,b^3_j\} \in E_b$ and~$\{a^3_j,w\} \in E_r$ form a \pt~in~$G_j-S$ which contradicts the fact that~$G-S$ is \ptt-free.
Next, consider~$\{a_j^p,b_j^3\}$. Clearly, $\{a_j^p,b_j^3\}$ is an element of~$S$, since otherwise~$\{a_j^p,b_j^3\} \in E_r$ and~$\{a_j^3,b_j^3\} \in E_b$ form a \pt~in~$G_j-S$ which contradicts the fact that~$G-S$ is \ptt-free. 

It remains to show that $E_{G_j}(\{a_j^p\},(B_j \setminus \{b_j^3\}) \cup W_j) \subseteq S$. Assume towards a contradiction that there exists an edge~$\{a_j^p, w\}$ in~$G_j-S$ with~$w \in (B_j \cup W_j) \setminus \{b_j^3\}$. If~$w = b_j^p$, the edges~$\{a_j^p, b_j^p\} \in E_b$ and~$\{b_j^p, a_j^3 \} \in E_r$ form a \pt~in~$G_j-S$. Otherwise, if~$w \neq b_j^p$ the edges~$\{a_j^p, w\} \in E_r$ and~$\{w, b^3_j\} \in E_b$ form a \pt~in~$G_j-S$. Both cases contradict the fact that~$G-S$ is \ptt-free and therefore~$E_{G_j-S}(\{a_j^p\},B_j \cup W_j) = \emptyset$ as claimed.~$\hfill \Diamond$
\end{claimproof}

We conclude from Claim \ref{Claim: NP-h all edges deleted} that deleting the~$14$ edges in~$E_{G_j}(\{a_j^p,a_j^q\},B_j\cup W_j)$ for distinct~$p,q \in \{1,2,3\}$ are the only three possible ways to destroy all \pt s in~$G_j$ with at most~$14$ edge deletions.

This fact combined with the fact that we need at least~$4$ edge deletions per variable gadget and~$|S|\leq 4 \cdot |X|+14 \cdot |\cC|$ implies that~$|E_G(A_j\cup B_j \cup W_j) \cap S|=14$ for each clause~$C_j$ and $|E_G(\{v_i\} \cup T_i \cup F_i) \cap S|=4$ for each variable~$x_i$.

We now define a satisfying assignment~$\mathcal{A}: X \rightarrow \{\text{true},\text{false}\}$ for~$\phi$ by
\begin{align*}
\mathcal{A}(x_i) := \begin{cases}
\text{true} &\text{if }E(\{v_i\},T_i) \subseteq S\text{, and}\\
\text{false}&\text{if }E(\{v_i\},F_i) \subseteq S\text{.}\\
\end{cases}
\end{align*}
The assignment~$\mathcal{A}$ is well-defined since in each variable gadget either all red or all blue edges belong to~$S$. 

Let~$C_j \in \cC$ be some clause in~$\phi$. It remains to show that~$C_j$ is satisfied by~$\mathcal{A}$. Since~$|E_G(A_j\cup B_j \cup W_j) \cap S|=14$ there are distinct~$p,q\in\{1,2,3\}$ such that~$S \cap E_G(A_j\cup B_j \cup W_j) = E_G(\{a_j^p,a_j^q\} , B_j\cup W_j)$. Without loss of generality assume~$p=1$ and~$q=2$. Therefore~$E_G(\{a_j^3\},W_j \cup B_j) \cap S = \emptyset$. By the construction of~$G$ we know that there is a variable~$x_i \in X$ occurring in~$C_j$ such that either~$a^3_j=t_i^{\Psi(C_j,x_i)}$ or~$a^3_j=f_i^{\Psi(C_j,x_i)}$.
We show that clause~$C_j$ is satisfied by the assignment~$\mathcal{A}(x_i)$.

For all~$w \in W_j$, $G-S$ contains the red edge~$\{a^3_j,w\}$. Moreover,~$G-S$ contains the blue edge~$\{a^3_j,b_j^3\}$. Since no vertex in~$W_j \cup B_j$ is adjacent to~$v_i$ we conclude that~$\{a^3_j, v_i\} \in S$ since otherwise~$\{a^3_j, v_i\}$ is part of a \pt~in~$G'$ which contradicts the fact that~$G'$ is \ptt-free.
If~$a^3_j=t_i^{\Psi(C_j,x_i)}$, then variable~$x_i$ occurs as a positive literal in~$C_j$ by the construction of~$G$. Then, $\{a^3_j, v_i\} \in E(\{v_i\},T_i)$. We conclude from~$\{a^3_j, v_i\} \in S$ that~$E(\{v_i\},T_i) \subseteq S$ and therefore~$\mathcal{A}(x_i)=\text{true}$. Analogously, if~$a^3_j=f_i^{\Psi(C_j,x_i)}$, then~$\{a^3_j, v_i\} \in E(\{v_i\},F_i)$. From~$\{a^3_j, v_i\} \in S$ we conclude~$\mathcal{A}(x_i)= \text{false}$.
In both cases, the assignment~$\mathcal{A}$ satisfies~$C_j$ which completes the correctness proof. 
%
%
%
\end{proof}

Note that for any instance~$\phi$ of \textsc{(3,4)-SAT} it holds that~$|\cC| \leq \frac{4}{3} |X|$. Thus, in the proof of Theorem~\ref{Theorem: NP-h} we constructed a graph with~$8 \cdot |X|+42 \cdot |\cC| \in \mathcal{O}(|X|)$ edges, $k=4 \cdot |X|+14 \cdot |\cC| \in \mathcal{O}(|X|)$, and therefore~$\ell = 4 \cdot |X| +28 \cdot |\cC| \in \mathcal{O}(|X|)$ for the dual parameter~$\ell=m -k$. Considering the ETH~\cite{IPZ01} and the fact that there is a reduction from \textsc{3-SAT} to \textsc{(3,4)-SAT} with a linear blow-up in the number of variables~\cite{Tovey84} this implies the following.

\begin{corollary} \label{Corollary: ETH lower bounds}
If the ETH is true, then \BPDs{} cannot be solved in~$2^{o(n+m+k+\ell)}$ time even if the maximum degree in~$G$ is~$8$.
\end{corollary}

\section{Polynomial-Time Solvable Cases} \label{Section: polytime}

Since \BPDs~is NP-hard, there is little hope to find a polynomial-time algorithm that solves \BPDs~on arbitrary instances. In this section we provide polynomial-time algorithms for two special cases of \BPDs~that are characterized by colored forbidden induced subgraphs.

\subsection{\BPDs~on Bicolored~$\mathbf{K_3}$-free Graphs} Our first result is a polynomial-time algorithm for \BPDs, when~$G=(V,E_r,E_b)$ does not contain a certain type of bicolored~$K_3$s.
\begin{defi}
  Three vertices~$u,v,w$ form a \emph{bicolored~$K_3$} if~$G[\{u,v,w\}]$ contains exactly three edges such that exactly two of them have the same color.  A bicolored~$K_3$ is \emph{endangered in~$G$} if at least one of the two edges with the
  same color is part of a \pt~in~$G$.
\end{defi}

A bicolored~$K_3$ on vertices~$u,v,w$ can be seen as an induced subgraph of~$G$, such that after one edge deletion in~$E_G(\{u,v,w\})$ one might end up with a new~\pt{} containing the vertices~$u,v$, and~$w$. This happens, if we delete one of the two edges with the same color. If the bicolored~$K_3$ is endangered, it might be necessary to delete one of these two edges to transform~$G$ into a \ptt-free graph. Intuitively, a graph~$G$ that contains no (endangered) bicolored~$K_3$ can be seen as a graph from which we can delete any edge that is part of a \pt{} without producing a new one. Note that the following result also implies that \BPDs{} can be solved in polynomial time on triangle-free graphs and thus also on bipartite graphs.

%
%

\begin{theorem}\label{thm:endangered}
  \BPDs{} can be solved in~$\Oh(nm^{\frac{3}{2}})$ time if~$G$ contains no endangered bicolored~$K_3$. 
\end{theorem}
\begin{proof}
We prove the theorem by reducing \BPDs~to \textsc{Vertex Cover} on bipartite graphs which can be solved in polynomial time since it is equivalent to computing a maximum matching.

\begin{quote}
  \textsc{Vertex Cover}\\
  \textbf{Input:} A graph~$G=(V, E)$ and an integer~$k \in \mathds{N}$.\\
  \textbf{Question:} Is there a \emph{vertex cover} of size at most~$k$ in~$G$, that is, a set~$S\subseteq V$ with~$|S| \leq k$ such that every edge~$e \in E$ has at least one endpoint in~$S$?
\end{quote}

 Let~$(G=(V,E_b,E_r),k)$ be an instance of \BPDs~where~$G$ contains no endangered bicolored~$K_3$. We define an instance~$(G',k')$ of \textsc{Vertex Cover} as follows. Let~$G'=(V',E')$ be the graph with vertex set~$V' := E_r \cup E_b$ and edge set~$E' := \{\{e_1,e_2\} \subseteq E_b\cup E_r \mid e_1 \text{ and } e_2 \text{ form a \pt~in }G \}$. That is,~$G'$ contains a vertex for each edge of~$G$ and edges are adjacent if they form a~$P_3$ in~$G$. Moreover, let~$k'=k$. The graph~$G'$ is obviously bipartite with partite sets~$E_b$ and~$E_r$.

We now show that~$(G,k)$ is a yes-instance for \textsc{BPD} if and only if~$(G',k')$ is a yes-instance for \textsc{Vertex Cover}.

$(\Rightarrow)$ Let~$S$ be a solution for~$(G,k)$. Note that the edges~of~$G$ are vertices of~$G'$ by construction and therefore~$S \subseteq V'$. We show that~$S$ is a vertex cover in~$G'$.
Assume towards a contradiction that there is an edge~$\{ x, y \} \in E'$ with~$x,y \not \in S$. By the definition of~$E'$, the edges~$x$ and~$y$ form a \pt~in~$G$. This contradicts the fact that~$G-S$ is \ptt-free. Hence, $S$ is a vertex cover of size at most~$k$ in~$G'$.

$(\Leftarrow)$ Let~$C \subseteq V'$ with~$|C| \leq k$ be a minimal vertex cover of~$G'$. Note that the vertices of~$G'$ are edges of~$G$ by construction and therefore~$C \subseteq E$. We show that~$G-C$ is \ptt-free.
 Assume towards a contradiction that there are edges~$x=\{u,v\}\in E_b \setminus C$ and~$y=\{v,w\}\in E_r \setminus C$ forming a \pt~in~$G-C$. Then,~$x$ and~$y$ do not form a \pt~in~$G$ since otherwise there is an edge~$\{x,y\} \in E'$, which has no endpoint in the vertex cover~$C$. It follows that there is an edge~$\{u,w\}$ in~$G$ that is not present in~$G-C$. Consequently,~$\{u,w\} \in C$. 
  Obviously, the vertices~$u,v,w$ form a bicolored~$K_3$. 
Since~$x$ and~$y$ form a \pt~in~$G-C$, one of these edges has the same color as~$\{u,w\}$.
 Since~$\{u,w\} \in C$ and~$C$ is minimal, it follows that~$\{u,w\} \in V'$ is an endpoint of an edge in~$G'$ and thus~$\{u,w\}$ is part of a \pt~in~$G$. Therefore,~$G[\{ u,v,w\}]$ forms an endangered bicolored~$K_3$ in~$G$ which contradicts the fact that~$G$ contains no endangered bicolored~$K_3$. This proves the correctness of the reduction.


For a given instance~$(G,k)$ of \BPDs, the \textsc{Vertex Cover} instance~$(G',k')$ can be computed in~$\mathcal{O}(nm)$ time by computing all \pt s of~$G$. Since \textsc{Vertex Cover} can be solved in~$\mathcal{O}(|E'| \cdot \sqrt{|V'|})$ time on bipartite graphs~\cite{HK73} and since $|V'|=m$ and~$|E'|\leq nm$, we conclude that \BPDs{} can be solved in~$\mathcal{O}(nm^\frac{3}{2})$ time on graphs without endangered~$K_3$s.
\end{proof}

\subsection{BPD on Graphs without Monochromatic~$K_3$s and~$P_3$s} We now show a second polynomial-time solvable special case that is characterized by four colored forbidden induced subgraphs: the two \emph{monochromatic}~$K_3$s, these are the~$K_3$s where all three edges have the same color, and the two \emph{monochromatic}~$P_3$s, these are the~$P_3$s where both edges have the same color. Observe that a graph that does not contain these forbidden induced subgraphs may still contain~$K_3$s or~$P_3$s. 

We provide two reduction rules that lead to a polynomial-time algorithm for this special case. These rules can also be applied to general instances of~\BPDs~and thus their running time bound is given for general graphs. We will later show that on graphs without monochromatic~$K_3$s and~$P_3$s we can apply them exhaustively in~$\Oh(n)$ time.

\begin{reduc}
\label{reduc:small-components}

\begin{enumerate}[a)]
\item Remove all \ptt-free components from~$G$.
\item If~$G$ contains a connected component~$C$ of
size at most five, then compute the minimum number of edge deletions~$k_C$ to
make~$G[C]$~\ptt-free, remove~$C$ from~$G$, and set~$k\leftarrow k-k_C$.
\end{enumerate}
\end{reduc}
\begin{lemma}
Reduction Rule~\ref{reduc:small-components} is correct and can be exhaustively applied in~$\Oh(nm)$~time.
\end{lemma}
\begin{proof}
The correctness of Reduction Rule~\ref{reduc:small-components} is obvious.  Part a) can be exhaustively applied in~$\Oh(nm)$ time as follows: First compute the connected components of~$G$ in~$\Oh(n+m)$ via breadth-first search. Then, enumerate all \pt s and label the vertices that are not part of any~\pt{} in~$\Oh(nm)$ time. Finally, remove all connected components without labeled vertices~in~$\Oh(n+m)$ time. Part b) of Reduction Rule~\ref{reduc:small-components} can be applied exhaustively in~$\Oh(n)$ time since we only need to find and remove connected components of constant size.
\end{proof}

The second reduction rule involves certain bridges that may be deleted greedily. An edge~$e$ is a \emph{bridge} if the graph~$G-e$ has more connected components than~$G$.
\begin{reduc}
\label{reduc:bridge-conflict}
\begin{enumerate}[a)]
 \item Remove all bridges from~$G$ that are not contained in any \pt{}. 
 \item
   If~$G$
  contains a bridge~$\{u,v\}$ such that\begin{itemize}
\item the connected component~$C$ containing~$v$ in~$G-\{u,v\}$ is~\ptt-free and
\item $\{u,v\}$ forms a~\pt{} with some edge~$\{v,w\}$ of~$C$ in~$G$,
\end{itemize}
then remove~$C$ from~$G$ and set~$k\leftarrow k-1$.
\end{enumerate}
\end{reduc}
\begin{lemma}
  Reduction Rule~\ref{reduc:bridge-conflict} is correct and can be applied exhaustively in $\Oh(n m)$~time.
\end{lemma}
\begin{proof}
  Let~$(G,k)$ be the original instance and~$(G',k')$ be the instance after the application
  of the rule. We first show the correctness of the two parts of the rule.
  For part~a), observe that for every subgraph~$G^*$ of~$G$, any bridge~$\{u,v\}$ that is removed by the rule is not part of a~\pt{}. Thus, any solution for~$G$ is a solution for~$G-\{u,v\}$ and vice versa. For part b), observe first that
  if~$(G',k')$ has a solution~$S'$, then~$S'\cup \{u,v\}$ is
  a solution for~$(G,k)$ and thus~$(G,k)$ is also a yes-instance. It remains to show that
  if~$(G,k)$ is a yes-instance, then so is~$(G',k')$. Consider a solution~$S$ with~$|S|\le k$
  for~$G$. Observe that~$\{u,v\}\in S$ or~$\{v,w\}\in S$. This
  implies~$S':=S\cap E(G[V\setminus C])\le k-1$. Finally, observe that
  since~$G[V\setminus C]=G'$,~$S'$ is a solution of size at most~$k-1$ for~$G'$.

  The running time can be seen as follows: We compute in~$\Oh(n+m)$ time the bridges of~$G$~\cite{Tar74}. Given the bridges, one can compute in~$\Oh(n+m)$ the block-cut-forest~$F$ of~$G$. The vertices of~$F$ are maximal 2-edge-connected components of~$G$ and the edges correspond to the bridges of~$G$. Then in $\Oh(nm)$~time, we enumerate all \pt{}s. Using the set of \pt{}s, we can compute for each 2-edge-connected component whether it contains a \pt{}. Moreover, for each bridge~$e$ and each incident 2-edge-connected component~$C$, we can compute whether~$e$ forms a \pt{} with some edge of~$C$. Finally, we can compute for each bridge the set of edges with which it forms a conflict. This additional information can be computed in~$\Oh(nm)$~time. We incorporate this information into the block-cut-forest~$F$ as follows: A vertex of~$F$ is colored black if the corresponding 2-edge-connected component contains a \pt{}, otherwise it is colored white.

  The exhaustive application of the rule is now performed on the block-cut-forest~$F$ via
  the following algorithm.  First, remove all white singletons from~$G$ and~$F$. Then,
  remove all bridges of~$G$ from~$G$ and~$F$ that are not part of in any~\pt{}
  in~$G$. Checking whether one of these two conditions is fulfilled can be performed
  in~$\Oh(n+m)$ time per removed edge and vertex. In the following, we assume that these
  removals have been applied
  exhaustively. 
  
  To describe the final part of the reduction, we denote for each vertex~$v$ of~$G$ the
  2-edge-connected component of~$G$ containing~$v$ by~$[v]$. Note that~$[v]$ is a vertex
  of~$F$. Now we check for each vertex~$v$ of~$G$ whether it is incident with a
  bridge~$\{u,v\}$ that fulfills the condition of part~b) of the rule.

  To do this efficiently, we characterize such bridges as follows.

  \begin{claim} \label{Claim:bridge-rule} A bridge~$\{u,v\}$ and a vertex~$v$ fulfill the
    condition of Reduction~Rule~\ref{reduc:bridge-conflict}~b) if and only if
    \begin{itemize}
    \item $\{u,v\}$ forms a~\pt{} with some edge of~$[v]$
      or some bridge that is incident with~$v$,
    \item 
      $[v]$ is white,
    \item $N_F([v])\setminus \{[u]\}$
    contains only white leaf-vertices, and 
  \item all other bridges incident with some vertex
    of~$[v]$ form a \pt{} with~$\{u,v\}$ and are not part of any further \pt{}.
  \end{itemize}

  \end{claim}
  \begin{claimproof}
    Let~$C$ denote the connected component containing~$v$ after the deletion of~$\{u,v\}$.
    Clearly, the stated conditions are sufficient for~$\{u,v\}$ and~$v$ to fulfill the
    requirements of the rule. For most conditions it is also clear that they are necessary
    for~$C$ to be~\pt-free. The only non-obvious condition is
    that~$N_F([v])\setminus \{[u]\}$ contains only leaf-vertices. This condition is
    necessary since, otherwise,~$C$ would contain some bridge~$e$ that is not incident with
    a vertex of~$[v]$. This bridge is part of some~\pt{}, since it has not been
    removed previously. This~\pt{} would also be contained in~$C$ since~$e$ and~$\{u,v\}$
    are not incident. $\hfill \Diamond$
  \end{claimproof}
  This characterization gives us now a way to check in~$\Oh(n+m)$ time whether~$G$
  contains some bridge that fulfills the condition of
  Reduction~Rule~\ref{reduc:bridge-conflict}~b). This check is done as follows. We
  consider each vertex~$x$ of~$F$. At most two bridges incident with some vertex of~$x$
  are candidates for fulfilling the conditions of the claim: all incident bridges must be
  incident with the same vertex~$v$ of~$x$ for the conditions to be fulfilled and an
  edge~$\{u,v\}$ must be the only incident bridge of its color if the conditions are
  fulfilled. For each candidate edge, we check in~$\Oh(|N_F([v])|)$ time whether the
  conditions of the claim are fulfilled using the precomputed information about the
  conflicts in~$G$. Thus, the total running time for this final check is linear in the
  number of edges of~$F$ and thus in~$\Oh(n+m)$. Note that after removing some
  edge~$\{u,v\}$ in this way, the remaining vertices of~$C$ are removed via the previous
  two checks.
  
  Altogether, applying the rules using~$F$ needs~$\Oh(n+m)$~time per removed vertex and
  edge. Since each application removes some vertex or bridge of~$G$, there are in
  total~$\Oh(n)$~applications. Moreover, after removing all isolated vertices from~$G$, we
  have~$n\le 2m$ and thus the overall running time of~$\Oh(nm)$ for the application of
  the rule follows.
\end{proof}

As we will show later, any graph without monochromatic $P_3$s and monochromatic~$K_3$s to which
the above reduction rules do not apply has maximum degree two. These graphs can be solved
in linear time as we see in the following lemma.
\begin{lemma}\label{lem:deg-two}
  Let~$(G,k)$ be an instance of~\BPDs{} such that~$G$ has maximum degree~$2$. Then,~$(G,k)$ can be solved in~$\Oh(n)$ time. 
\end{lemma}
\begin{proof}
  In the following, we construct a solution~$S$ for~$(G,k)$. If~$G$ has maximum degree~$2$, then each connected component of~$G$ is either a path or a cycle. The algorithm first deals with cycles and then considers the remaining paths. Observe that every connected component of size at most~$3$ can be solved within~$\Oh(1)$~time. For the rest of the proof we assume that every connected component has size at least~$4$.

  First, we consider each connected component~$C$ of~$G$ which is a cycle. We either transform~$C$ into one or two paths or solve~$C$ directly. First, assume that~$C$ contains three subsequent edges~$e_1, e_2$, and~$e_3$ of the same color, then edge~$e_2$ is not part of any \pt. Hence~$e_2$ can be removed without decreasing~$k$. The remaining connected component is a path and will be solved in the second step. Second, assume that~$C$ contains two subsequent edges~$e_1$ and~$e_2$ with the same color. Recall that we may assume that~$|C|\ge4$. Further, let~$e_0$ be the other edge that is incident with~$e_1$ and let~$e_3$ be the other edge that is incident with~$e_2$. According to our assumption,~$e_0$ and~$e_1$ form a \pt~and~$e_2$ and~$e_3$ form a \pt. Hence, either~$e_0\in S$ or~$e_1\in S$ and either~$e_2\in S$ or~$e_3\in S$. Since~$e_1$ and~$e_2$ have the same color and no further edges are incident with~$e_1$ and~$e_2$, we may assume that~$e_0, e_3\in S$. The remaining connected components are paths and will be solved in the second step. Third, consider the case that~$C$ contains no two subsequent edges of the same color. Then~$C$ consists of~$2\ell$ edges~$e_1, \ldots , e_{2\ell}$ and each two subsequent edges form a \pt. Thus,~$C$ contains a set of~$\ell$ edge-disjoint \pt s:~$\{e_1, e_2\}, \{ e_3, e_4\}, \ldots , \{ e_{2\ell-1}, e_{2\ell}\}$ and contains exactly~$\ell$ blue edges. Thus, deleting the $\ell$~blue edges of~$C$ is optimal.

In a second step, we consider each connected component that is a path. Let~$P_n$ be a path consisting of~$n$ vertices~$v_1, \ldots , v_n$. Visit the edges~$\{v_i,v_{i+1}\}$ for increasing~$i$ starting at~$v_1$. For each edge, check whether it is part of some \pt{}. Let~$\{v_i,v_{i+1}\}$ be the first encountered edge that is in~a~\pt{}. Then, delete~$\{v_{i+1},v_{i+2}\}$, decrease~$k$ by one, and continue with~$\{v_{i+2},v_{i+3}\}$ if it exists. First, observe that~$\{v_{i+1},v_{i+2}\}$ exists since~$\{v_i,v_{i+1}\}$ does not form a~\pt{} with~$\{v_{i-1},v_{i}\}$. Second, observe that the deletion of~$\{v_{i+1},v_{i+2}\}$ is simply an application of Reduction~Rule~\ref{reduc:bridge-conflict} and therefore correct. Clearly, this greedy algorithm runs in~$\Oh(n)$ time. 

In altogether~$\Oh(n)$ time, we can consider each cycle~$C$ and either solve~$C$ or delete one or two edges, which transforms~$C$ into one or two paths.
The greedy algorithm for paths runs in~$\Oh(n)$ time on all paths. Thus, the remaining instance can be solved in~$\Oh(n)$ time. The overall running time follows.
\end{proof}

We have now all ingredients to present the polynomial-time algorithm for graphs without monochromatic~$K_3$ and monochromatic~$P_3$. In order to prove the correctness of the algorithm and the linear running time, we make the following observation about such graphs.
\begin{lemma}\label{lem:deg-bound}
  Let~$G$ be a graph that contains no monochromatic~$K_3$ and no monochromatic~$P_3$ as
  induced subgraphs. Then, the maximum blue degree and the maximum red degree in~$G$ are
  2.
\end{lemma}
\begin{proof}
  We show the proof only for the blue degree, the bound for the red degree can be shown
  symmetrically.
  Assume towards a contradiction that~$G$ contains a vertex~$t$ with at least three blue
  neighbors~$u$,~$v$, and~$w$. Since~$G$ contains no blue~$P_3$, the
  subgraph~$G[\{u,v,w\}]$ has three edges. Moreover, since~$G$ contains no
  monochromatic~$K_3$ not all of the three edges in~$G[\{u,v,w\}]$ are red. Assume without
  loss of generality that~$\{u,v\}$ is blue. Then~$G[\{u,v,t\}]$
  is a blue~$K_3$, a contradiction.
\end{proof}

\begin{figure}[t]
  \begin{center}
    \begin{tikzpicture}[scale=0.6]
      \tikzstyle{knoten}=[circle,fill=white,draw=black,minimum size=5pt,inner sep=0pt]

      \begin{scope}

        \node[knoten] (1) at (0,0) {};
        \node[knoten] (2) at (2,0) {};
        \node[knoten] (3) at (2,2) {};
        \node[knoten] (4) at (0,2) {};
        \node at (1,-0.7) {paw};
        
        \draw[-] (1) edge (2);
        \draw[-] (1) edge (3);
        \draw[-] (2) edge (3);
        \draw[-] (1) edge (4);
        
    \begin{scope}[xshift=4cm]
        \node[knoten] (1) at (0,0) {};
        \node[knoten] (2) at (2,0) {};
        \node[knoten] (3) at (2,2) {};
        \node[knoten] (4) at (0,2) {};
        \node at (1,-0.7) {claw};
        
        \draw[-] (1) edge (2);
        \draw[-] (1) edge (3);
        \draw[-] (1) edge (4);
        
    \end{scope}

    \begin{scope}[xshift=8cm]
        \node[knoten] (1) at (0,0) {};
        \node[knoten] (2) at (2,0) {};
        \node[knoten] (3) at (2,2) {};
        \node[knoten] (4) at (0,2) {};
        \node at (1,-0.7) {$K_4$};
        
        \draw[-] (1) edge (2);
        \draw[-] (1) edge (3);
        \draw[-] (1) edge (4);
        \draw[-] (3) edge (2);
        \draw[-] (2) edge (4);
        \draw[-] (3) edge (4);
        
    \end{scope}
    
    \begin{scope}[xshift=12cm]
        \node[knoten] (1) at (0,0) {};
        \node[knoten] (2) at (2,0) {};
        \node[knoten] (3) at (2,2) {};
        \node[knoten] (4) at (0,2) {};
        \node at (1,-0.7) {diamond};
        
        \draw[-] (1) edge (2);
        \draw[-] (1) edge (3);
        \draw[-] (1) edge (4);
        \draw[-] (2) edge (3);
        \draw[-] (3) edge (4);
        
    \end{scope}

  \end{scope}        
\end{tikzpicture}
\end{center}
\caption{The (uncolored) small graphs used in this work.}
\label{Figure: claw and paw}
\end{figure}

\begin{theorem}
\label{theo:no monochromatic subgraph linear time}
  \BPDs{} can be solved in~$\Oh(n)$ time if~$G$ contains no monochromatic~$K_3$ and no monochromatic~$P_3$.
\end{theorem}
\begin{proof}
  The algorithm first applies Reduction Rule~\ref{reduc:small-components} exhaustively. Afterwards, Reduction Rule~\ref{reduc:bridge-conflict} is applied on all bridges~$\{u,v\}$ with~$\deg(u)\geq 3$ or~$\deg(v) \geq 3$. Thus, let~$G$ be the graph after the applications of Reduction Rules~\ref{reduc:small-components} and~\ref{reduc:bridge-conflict} as described above. We show that~$G$ has maximum degree at most~$2$. Afterwards, Lemma~\ref{lem:deg-two} applies and the remaining instance can be solved in~$\Oh(n)$ time. Observe that by
  Lemma~\ref{lem:deg-bound}, the maximum degree in~$G$ is~$4$. For an illustration of the small uncolored graphs used in this proof see Figure~\ref{Figure: claw and paw}.

  First, assume that the maximum degree of~$G$ is~$4$ and let~$v$ be a vertex of degree~$4$. 
  We show that~$N[v]$ is a connected component of~$G$. This implies that~$N[v]$ is removed by Reduction
  Rule~\ref{reduc:small-components} in this case. By Lemma~\ref{lem:deg-bound},~the vertex~$v$ has exactly two
  blue neighbors~$u_1$,~$u_2$ and exactly two red neighbors~$w_1$,~$w_2$. Since~$G$ contains no
  monochromatic~$P_3$ and no monochromatic~$K_3$,~$\{ u_1, u_2\}$ is red and~$\{ w_1,w_2\}$ is blue. Now assume towards a
  contradiction, that one of these four vertices has a neighbor~$t\notin N[v]$
  in~$G$.

\begin{figure}
\begin{center}
\begin{tikzpicture}[scale=0.7]
\tikzstyle{knoten}=[circle,fill=white,draw=black,minimum size=5pt,inner sep=0pt]
\tikzstyle{bez}=[inner sep=0pt]
		\node[knoten] (v)[label=above:{$v$}] at (1.6,1) {};
		\node[knoten] (u1)[label=left:{$u_1$}] at (0,0) {};
		\node[knoten] (u2)[label=right:{$u_2$}] at (1,0) {};
		\node[knoten] (w1)[label=below:{$w_1$}] at (2.2,0) {};
		\node[knoten] (w2)[label=below:{$w_2$}] at (3.2,0) {};
		\node[knoten] (t)[label=below:{$t$}] at (0.5,-1) {};
		\draw[-,color=blue,very thick] (v) to (u1);
		\draw[-,color=blue,very thick] (v) to (u2);
		\draw[-,color=blue,very thick] (t) to (u2);
		\draw[-,color=blue,very thick] (w1) to (w2);
		\draw[-,color=myred] (v) to (w1);
		\draw[-,color=myred] (v) to (w2);
		\draw[-,color=myred] (u1) to (u2);
		\draw[-,color=myred] (u1) to (t);
		
		\node[knoten] (t)[label=above:{$t$}] at (6.75,0.75) {};
		\node[knoten] (v)[label=left:{$v$}] at (6,0) {};
		\node[knoten] (u)[label=right:{$u$}] at (7.5,0) {};
		\node[knoten] (w)[label=270:{$w$}] at (6.75,-0.75) {};
		\draw[-,color=blue,very thick] (v) to (u);
		\draw[-,color=blue,very thick] (v) to (t);
		\draw[-,color=blue,very thick] (w) to (u);
		\draw[-,color=myred] (u) to (t);
		\draw[-,color=myred] (v) to (w);
		
		\node[knoten] (u)[label=left:{$u$}] at (10,0.75) {};
		\node[knoten] (v)[label=above:{$v$}] at (11.5,0) {};
		\node[knoten] (w)[label=left:{$w$}] at (10,-0.75) {};
		\node[knoten] (t)[label=above:{$t$}] at (12.5,0) {};
		\draw[-,color=blue,very thick] (v) to (u);
		\draw[-,color=blue,very thick] (v) to (w);
		\draw[-,color=myred] (v) to (t);
		\draw[-,color=myred] (u) to (w);


		\node[bez] (a) at (-1,1.5) {(a)};
		\node[bez] (b) at (5,1.5) {(b)};
		\node[bez] (c) at (8.7,1.5) {(c)};
\end{tikzpicture}
\end{center}
\caption{Subgraphs used in the proof of Theorem~\ref{theo:no monochromatic subgraph linear time} (a)~$|N(v)|=4$, (b)~$G[N[v]]$ is a diamond and (c)~$G[N[v]]$ is a paw.} \label{Figure: Proof monochromatic in linear time}
\end{figure}
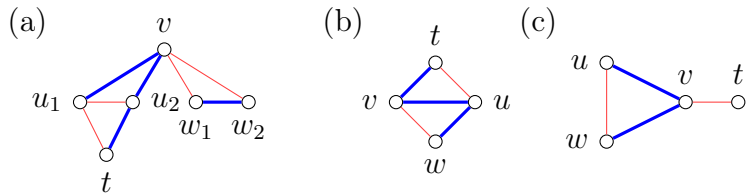  
   Without loss of generality assume that this vertex is~$u_1$. See Figure~\ref{Figure: Proof monochromatic in linear time} (a) for an example. Then,~$\{ u_1,t\}$
  is red because otherwise~$G[\{ u_1,v,t\}]$ is a
  blue~$P_3$. This implies that~$\{ u_2,t\}$ is blue because
  otherwise~$G[\{ u_1,u_2,t\}]$ is a red~$P_3$ or a red~$K_3$. Then,
  however,~$G[\{ u_2,v,t\}]$ is a blue~$P_3$, a contradiction. Altogether, this implies that~$N[v]$ is
  a connected component of~$G$.
  Hence, if~$G$ is reduced with respect to Reduction~Rule~\ref{reduc:small-components},
  then~$G$ contains no vertices of degree~$4$.

  Second, assume that the maximum degree of~$G$ is~$3$ and let~$v$ be a vertex of degree~$3$. 
  If~$G[N[v]]$ is a~$K_4$, then $N[v]$ is a connected component of~$G$ and is
  removed by Reduction~Rule~\ref{reduc:small-components}. Next, assume~$G[N[v]]$ is a
  diamond, and let~$t$, $u$, and~$w$ denote the neighbors of~$v$ where~$u$ is the other
  vertex that has degree three in~$G[N[v]]$. See Figure~\ref{Figure: Proof monochromatic in linear time} (b) for an example. Assume without loss of generality that~$u$ is
  a blue neighbor of~$v$. Then, by Lemma~\ref{lem:deg-bound}, one of~$t$ and~$w$, say~$w$, is a red neighbor of~$v$. This
  implies that~$t$ is a blue neighbor of~$v$, because otherwise~$G[\{ v,w,t\}]$ is a
  red~$P_3$. Consequently,~$t$ is a red neighbor of~$u$ because otherwise~$t$,~$v$,
  and~$u$ form a blue~$K_3$. Hence,~$w$ is a blue neighbor of~$u$ because
  otherwise~$G[\{ u,w,t\}]$ is a red~$P_3$. Altogether, we have that~$t$ and~$w$ each
  have a red and a blue neighbor in~$\{u,v\}$. Since~$N[v]=N[u]$, we have that~$t$ and~$w$
  cannot have a neighbor in~$V\setminus N[v]$ as such a neighbor would form a
  monochromatic~$P_3$ with one of~$u$ and~$v$. Hence,~$N[v]$ is a connected component
  of~$G$ in this case and consequently removed by Reduction
  Rule~\ref{reduc:small-components}.

  Thus~$G[N[v]]$ is neither a~$K_4$ nor a diamond if~$G$ is reduced with respect
  to~Reduction Rule~\ref{reduc:small-components}. Moreover,~$G[N[v]]$ cannot
  be a claw since in this case~$G$ contains a monochromatic~$P_3$. Hence, the only
  remaining case is that~$G[N[v]]$ is a paw. Let~$t$,~$u$, and~$w$ be the neighbors of~$v$
  where~$u$ and~$w$ are adjacent. For an example see Figure~\ref{Figure: Proof monochromatic in linear time}~(c). Assume furthermore without loss of generality that~$v$
  is incident with two blue edges. This implies that~$t$ is a red neighbor of~$v$ as otherwise~$t$, $v$, and~$u$ form a monochromatic~$P_3$. Also,~$u$
  and~$w$ are blue neighbors of~$v$. Consequently,~$\{ u,w\}$ is red. As in the proof above for the case that~$G[N[v]]$ is a diamond,~$u$ and~$w$ have
  no further neighbor in~$G$. Thus,~$\{v,t\}$ is a bridge with~$\deg(v) = 3$ that fulfills the condition of
  Reduction Rule~\ref{reduc:bridge-conflict}~b) and thus~$\{v,t\}$ is removed by this
  rule. Altogether this implies that any instance to which
  Reduction Rules~\ref{reduc:small-components} and~\ref{reduc:bridge-conflict} have been applied as described above has maximum degree~$2$. 
  By Lemma~\ref{lem:deg-two}, we can thus solve the remaining instance in linear time.

Next, we consider the running times of Reduction Rules~\ref{reduc:small-components} and~\ref{reduc:bridge-conflict} in more detail since for both rules the running time analysis given above did \emph{not} assume that~$G$ contains no monochromatic~$P_3$ and no monochromatic~$K_3$.

First, we apply Reduction Rule~\ref{reduc:small-components} exhaustively. Since~$G$ has maximum degree at most four, we can label all vertices that are part of some \pt{} in~$\Oh(n)$ time and thus Reduction Rule~\ref{reduc:small-components} can be applied exhaustively in~$\Oh(n)$ time. Observe that in the resulting graph the maximum degree of~$G$ is three since vertices of degree~four are in connected components of size~five.

Next, we consider the running time of Reduction Rule~\ref{reduc:bridge-conflict}, after Reduction Rule~\ref{reduc:small-components} was applied exhaustively. Recall that Reduction Rule~\ref{reduc:bridge-conflict} is only applied to bridges that have at least one endpoint with degree three. To apply the rule exhaustively, we first compute in~$\Oh(n)$ time the set of all vertices of degree three. For each such vertex~$v$, the graph~$G[N[v]]$ is a paw because otherwise,~$N[v]$ is a connected component of constant size as shown above. Hence, there exist two vertices~$u, w\in N(v)$ such that~$N(u)=\{ v,w\}$ and~$N(w)=\{ u,v\}$. The vertices~$u$ and~$w$ can be determined in~$\Oh(1)$ time. Let~$t\in N(v)\setminus\{ u,w\}$. Then, Reduction Rule~\ref{reduc:bridge-conflict} removes~$\{v,t\}$ from~$G$. Thus, in~$\Oh(1)$ time, we may apply Reduction Rule~\ref{reduc:bridge-conflict} on the bridge containing~$v$. Consequently, the rule can be applied exhaustively on all degree-three vertices~$\Oh(n)$ time. Afterwards, the remaining instance has maximum degree two and can be solved in~$\Oh(n)$ time. Hence, the overall running time is $\Oh(n)$.
\end{proof}

\section{Parameterized Complexity} \label{Section: param}

In this section we study the parameterized complexity of \BPDs~parameterized by~$k$, $\ell:=m-k$, and~$(k,\Delta)$, where~$\Delta$ denotes the maximum degree of~$G$. We first provide an~$\Oh(1.84^k\cdot nm)$-time fixed-parameter algorithm for \BPDs. Afterwards, we study problem kernelizations for \BPDs~parameterized by~$(k,\Delta)$ and~$\ell$.


\subsection{ A Fixed-Parameter Algorithm for Bicolored~$\mathbf{P_3}$~Deletion}

We now provide a fixed-parameter algorithm that solves \BPDs~parameterized by~$k$. Note that there is a naive~$\Oh(2^k \cdot nm)$ branching algorithm for \BPDs: For a given instance~$(G,k)$, check in~$\Oh(nm)$ time if~$G$ contains a \pt. If this is not the case, then answer yes. Otherwise, answer no if~$k<1$. If~$k\ge 1$, then compute a \pt~formed by the edges~$e_1$ and~$e_2$ and branch into the cases~$(G-e_1,k-1)$ and~$(G-e_2,k-1)$. We modify this simple algorithm by branching on slightly more complex structures, obtaining a running time of~$\Oh(1.84^k \cdot nm)$. Note that by Corollary~\ref{Corollary: ETH lower bounds} a subexponential algorithm in~$k$ is not possible when assuming the ETH.

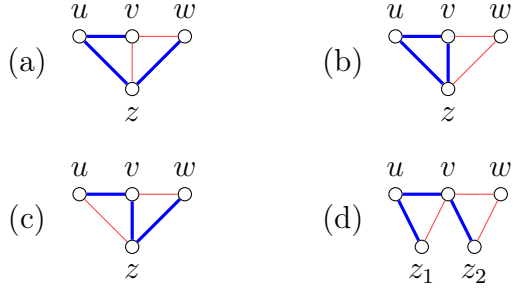
\begin{figure}
\begin{center}
\begin{tikzpicture}[scale=0.7]
\tikzstyle{knoten}=[circle,fill=white,draw=black,minimum size=5pt,inner sep=0pt]
\tikzstyle{bez}=[inner sep=0pt]
		\node[knoten] (u1)[label=above:{$u$}] at (-1,5) {};
		\node[knoten] (v1)[label=above:{$v$}] at (0,5) {};
		\node[knoten] (w1)[label=above:{$w$}] at (1,5) {};
		\node[knoten] (z1)[label=270:{$z$}] at (0,4) {};
		\draw[-,color=blue,very thick] (u1) to (v1);
		\draw[-,color=myred] (v1) to (w1);
		\draw[-,color=blue,very thick] (u1) to (z1);
		\draw[-,color=myred] (v1) to (z1);
		\draw[-,color=blue,very thick] (w1) to (z1);
		
		\node[knoten] (u2)[label=above:{$u$}] at (5,5) {};
		\node[knoten] (v2)[label=above:{$v$}] at (6,5) {};
		\node[knoten] (w2)[label=above:{$w$}] at (7,5) {};
		\node[knoten] (z2)[label=270:{$z$}] at (6,4) {};
		\draw[-,color=blue,very thick] (u2) to (v2);
		\draw[-,color=myred] (v2) to (w2);
		\draw[-,color=blue,very thick] (u2) to (z2);
		\draw[-,color=blue,very thick] (v2) to (z2);
		\draw[-,color=myred] (w2) to (z2);
		
		\node[knoten] (u3)[label=above:{$u$}] at (-1,2) {};
		\node[knoten] (v3)[label=above:{$v$}] at (0,2) {};
		\node[knoten] (w3)[label=above:{$w$}] at (1,2) {};
		\node[knoten] (z3)[label=270:{$z$}] at (0,1) {};
		\draw[-,color=blue,very thick] (u3) to (v3);
		\draw[-,color=myred] (v3) to (w3);
		\draw[-,color=myred] (u3) to (z3);
		\draw[-,color=blue,very thick] (v3) to (z3);
		\draw[-,color=blue,very thick] (w3) to (z3);
		
		\node[knoten] (u4)[label=above:{$u$}] at (5,2) {};
		\node[knoten] (v4)[label=above:{$v$}] at (6,2) {};
		\node[knoten] (w4)[label=above:{$w$}] at (7,2) {};
		\node[knoten] (z41)[label=270:{$z_1$}] at (5.5,1) {};
		\node[knoten] (z42)[label=270:{$z_2$}] at (6.5,1) {};
		\draw[-,color=blue,very thick] (u4) to (v4);
		\draw[-,color=myred] (v4) to (w4);
		\draw[-,color=blue,very thick] (u4) to (z41);
		\draw[-,color=myred] (v4) to (z41);
		\draw[-,color=blue,very thick] (v4) to (z42);
		\draw[-,color=myred] (w4) to (z42);


		\node[bez] (a) at (-2,4.5) {(a)};
		\node[bez] (b) at (4,4.5) {(b)};
		\node[bez] (a) at (-2,1.5) {(c)};
		\node[bez] (b) at (4,1.5) {(d)};
\end{tikzpicture}
\end{center}
\caption{The graphs used in the proof of Theorem~\ref{Theorem: Branching Algorithm} (a) LC-Diamond, (b) LO-Diamond, (c) IIZ-Diamond and (d) CC-Hourglass. The names refer to the shapes of the connected components of both colors. Due to symmetry we use the same names if we swap the red and blue colors of the edges in each graph.} \label{Figure: Branching Subgraphs}
\end{figure}

The basic idea of the algorithm is to branch on LC-Diamonds, LO-Diamonds, IIZ-Diamonds and CC-Hourglasses. For the definition of these structures see Figure~\ref{Figure: Branching Subgraphs}. We say that a graph~$G$ is \emph{nice} if~$G$ has none of the structures from Figure~\ref{Figure: Branching Subgraphs} as induced subgraph and every edge of~$G$ forms a \pt~with at most one other edge of~$G$. We give a polynomial-time algorithm that solves \BPDs~when the input graph is nice. To this end consider the following proposition.

\begin{proposition} \label{Proposition: No more conflicts after branching}
Let~$(G=(V,E_r,E_b),k)$ be an instance of~\BPDs{} such that~$G$ is nice. Moreover, let~$p$ be the number of~\pt s in~$G$. Then, for every two edges~$e_1$ and~$e_2$ forming a \pt~in~$G$ there is an edge~$e \in \{e_1,e_2\}$ such that
\begin{enumerate}
\item[a)] $G-e$ contains~$p-1$ \pt s and every~\pt~of~$G-e$ is a~\pt~in~$G$, and
\item[b)] $G-e$ is nice.
\end{enumerate}

\end{proposition}
\begin{proof} For the proof of Statement~\textit{a)}, let~$u$ and~$v$ denote the endpoints of~$e_1 \in E_b$ and let~$v$ and~$w$ denote the endpoints of~$e_2 \in E_r$. Note that the number of \pt s in~$G-e_1$ and~$G-e_2$ is at least~$p-1$ since every edge of~$G$ is part of at most one \pt~since~$G$ is nice. It remains to show that there is an edge~$e \in \{e_1, e_2\}$ such that the number of \pt s in~$G-e$ is at most~$p-1$, and that every~\pt~in~$G-e$ is a~\pt~in~$G$. Assume towards a contradiction that there are at least~$p$ \pt s in~$G-e_1$ and in~$G-e_2$. Then, there exist vertices~$z_1,z_2 \in V$, such that~$\{u,z_1\}$ and~$\{z_1,v\}$ form a \pt~in~$G-e_1$ and also~$\{v,z_2\}$ and~$\{z_2,w\}$ form a \pt~in~$G-e_2$.

First, assume~$z_1 = z_2 =: z$. If~$\{v,z\} \in E_r$, then~$\{u,z\} \in E_b$, $\{w,z\}\in E_b$, and~$G[\{u,v,w,z\}]$ is an LC-Diamond with two red and three blue edges. This contradicts the fact that~$G$ contains no induced LC-Diamond. Analogously, if~$\{v,z\} \in E_b$, then~$G[\{u,v,w,z\}]$ is an LC-Diamond with two blue and three red edges. We conclude~$z_1 \neq z_2$.

Second, assume~$\{v,z_1\} \in E_b$. Then, $\{u,z_1\} \in E_r$. Since every edge of~$G$ is part of at most one~\pt, there is an edge~$\{z_1,w\} \in E$. If~$\{z_1, w\} \in E_r$, then~$G[\{u,v,w,z_1\}]$ is an LC-Diamond. Otherwise, if~$\{z_1,w\} \in E_b$, then~$G[\{u,v,w,z_1\}]$ is an IIZ-Diamond. This contradicts the fact that~$G$ contains no induced LC- or IIZ-Diamond. Therefore,~$\{v,z_1\} \in E_r$. With the same arguments we can show that~$\{v,z_2\} \in E_b$.

From~$\{v,z_1\} \in E_r$ and~$\{v,z_2\} \in E_b$ we conclude that~$\{u,z_1\} \in E_b$ and~$\{w,z_2\} \in E_r$. Next, if we have~$E(\{u,v,w,z_1,z_2\})=\{ \{u,v\}, \{u,z_1\}, \{v,w\},\{v,z_1\},\{v,z_2\},\{w,z_2\}\}=: A$, then the graph~$G[\{u,v,w,z_1,z_2\}]$ is an induced CC-Hourglass which contradicts the fact that~$G$ does not contain induced CC-Hourglasses. Hence, $E(\{u,v,w,z_1,z_2\}) \supsetneq A$. Consider the following cases.

\textbf{Case 1:} $\{u,w\} \in E$. Then~$\{u,v\}$ and~$\{v,w\}$ do not form a \pt, which contradicts the choice of~$\{u,v\}$ and~$\{v,w\}$.

\textbf{Case 2:} $\{z_1,w\} \in E$. If~$\{z_1,w\} \in E_b$, then~$G[\{u,v,w,z_1\}]$ is an induced LC-Diamond. Otherwise, if~$\{z_1,w\} \in E_r$, then~$G[\{u,v,w,z_1\}]$ is an induced LO-Diamond. Both cases contradict the fact that~$G$ is nice.

\textbf{Case 3:} $\{z_2,u\} \in E$. If~$\{z_2,u\} \in E_r$, then~$G[\{u,v,w,z_2\}]$ is an induced LC-Diamond. Otherwise, if~$\{z_1,w\} \in E_b$, then~$G[\{u,v,w,z_2\}]$ is an induced LO-Diamond. Both cases contradict the fact that~$G$ is nice.

\textbf{Case 4:} $\{z_1,z_2\} \in E$. If~$\{z_1,z_2\} \in E_b$, then~$G[\{z_1,z_2,w,v\}]$ is an LC-Diamond. Otherwise, if~$\{z_1,z_2\} \in E_r$, then~$G[\{z_1,z_2,u,v\}]$ is an LC-Diamond. Both cases contradict the fact that~$G$ is nice.

All cases lead to a contradiction. Hence, there exists~$e \in \{e_1, e_2\}$ such that~$G-e$ contains~$p-1$ \pt s which proves Statement~\textit{a)}.

 Next, we show Statement~\textit{b)}. To this end, let~$e_1$ and~$e_2$ be two edges forming a bicolored~$P_3$ in~$G$. Let~$e \in \{e_1,e_2\}$ that satisfies~\textit{a)}. We show that~$G-e$ is nice. From~\textit{a)} we know that every~\pt{} of~$G-e$ is also a~\pt{} in~$G$. Hence, the fact that every edge of~$G$ is part of at most one \pt~implies that every edge of~$G-e$ is part of at most one \pt.

First, assume towards a contradiction that~$G-e$ contains an induced LC-, LO- or IIZ-Diamond~$(G-e)[\{u,v,w,z\}]$ as given in Figure~\ref{Figure: Branching Subgraphs}. Since~$G$ contains no such structure, we conclude~$e=\{u,w\}$ and $\{u,v,w,z\}$ is a clique in~$G$. Then, deleting~$e$ from~$G$ produces a new \pt~on edges~$\{u,v\}$ and~$\{v,w\}$ in~$G-e$ which contradicts Statement~\textit{a)}. Therefore,~$G-e$ contains no induced LC-, LO- and IIZ-Diamonds.

Second, assume towards a contradiction that the graph~$G-e$ contains some induced CC-Hour\-glass~$(G-e)[\{u,v,w,z_1,z_2\}]$ as given in Figure~\ref{Figure: Branching Subgraphs}. Then, since~$G$ does not contain an induced CC-Hourglass, both endpoints of~$e$ are elements of~$\{u,v,w,z_1,z_2\}$ and~$G[\{u,v,w,z_1,z_2\}]$ contains exactly seven edges.

\textbf{Case 1:} $e=\{u,w\}$ (or~$e=\{z_1,z_2\}$). Then,~$G-e$ contains the new \pt~formed by the edges~$\{u,v\}$ and~$\{u,w\}$ (by $\{z_1,v\}$ and $\{v,z_2\}$, respectively) which contradicts~\textit{a)}.

\textbf{Case 2:} $e=\{z_1,w\}$. Then, $G[\{u,v,w,z_1\}]$ is an induced LC-Diamond in~$G$ if~$e \in E_b$ or an induced LO-Diamond in~$G$ if~$e \in E_r$ which contradicts the fact that~$G$ has no induced LC-Diamonds and LO-Diamonds.

\textbf{Case 3:} $e=\{z_2,u\}$ Then, $G[\{u,v,w,z_2\}]$ is an induced LO-Diamond in~$G$ if~$e \in E_b$ or an induced LC-Diamond in~$G$ if~$e \in E_r$ which contradicts the fact that~$G$ has no induced LC-Diamonds and LO-Diamonds.

All cases lead to a contradiction and therefore~$G-e$ contains no induced LC-, LO-, IIZ-Diamonds, CC-Hourglasses and every edge of~$G-e$ is part of at most one \pt. 
\end{proof}

Proposition~\ref{Proposition: No more conflicts after branching} implies a simple algorithm for \BPDs~on such graphs.

\begin{corollary} \label{Corollary: Solve after Branching}
Let~$(G,k)$ be an instance of \BPDs~where~$G$ is nice. Then, we can decide in~$\Oh(nm)$ time whether~$(G,k)$ is a yes- or a no-instance of \BPDs.
\end{corollary}

\begin{proof}
We solve \BPDs~with the following algorithm: First, enumerate all \pt s in  $\Oh(nm)$~time. Second, check if there are at most~$k$ \pt s. If yes,~$(G,k)$ is a yes-instance. Otherwise,~$(G,k)$ is a no-instance.

It remains to show that this algorithm is correct. Assume~$G$ contains~$p$ \pt s. Since every edge of~$G$ forms a \pt~with at most one other edge, all \pt s in~$G$ are edge-disjoint. Hence,~$p$ edge deletions are necessary. By Proposition~\ref{Proposition: No more conflicts after branching}~\textit{a)} we can eliminate exactly one \pt~with one edge deletion without producing other \pt s. By Proposition~\ref{Proposition: No more conflicts after branching}~\textit{b)} this can be done successively with every \pt, since after deleting one of its edges we do not produce LC-, LO-, IIZ-Diamonds, CC-Hourglasses or edges that form a \pt~with more than one other edge. Thus,~$p$ edge deletions are sufficient. Hence, the algorithm is correct.
\end{proof}

Next, we describe how to transform an arbitrary graph~$G$ into a nice graph~$G'$ by branching. To this end consider the following branching rules applied on an instance~$(G,k)$ of \BPDs.

\begin{branch} \label{BR: eliminate multi pts}
If there are three distinct edges~$e_1, e_2, e_3 \in E_r \cup E_b$ such that~$e_1$ forms a \pt~with~$e_2$ and with $e_3$, then branch into the cases
\begin{enumerate}
\item[•] $I_1:=(G-e_1,k-1)$, and
\item[•] $I_2:=(G-\{e_2, e_3\},k-2)$.
\end{enumerate}
\end{branch}

\begin{lemma}
Branching Rule~\ref{BR: eliminate multi pts} is correct.
\end{lemma}

\begin{proof}
We show that~$(G,k)$ is a yes-instance of \BPDs~if and only if at least one of the instances~$I_1$ or~$I_2$ is a yes-instance of \BPDs.

$(\Leftarrow)$ Assume~$I_1$ is a yes-instance or~$I_2$ is a yes-instance. In each branching case~$I_i$, the parameter~$k$ is decreased by the exact amount~$p_i$ of edges deleted from~$G$. Therefore, if some~$I_i$ has a solution of size at most~$k-p_i$, then~$(G,k)$ has a solution.

$(\Rightarrow)$ Let~$S$ be a solution for~$G$. Since~$e_1$ and~$e_2$ form a \pt, at least one of these edges belongs to~$S$. If~$e_1 \in S$, then~$I_1$ is a yes-instance since we can transform~$G-e_1$ into a \ptt-free graph by deleting the at most~$k-1$ edges in~$S\setminus \{e_1\}$. Otherwise, if~$e_1 \not \in S$, then~$e_2, e_3 \in S$. Hence, $I_2$ is a yes-instance since we can transform~$G-\{e_2,e_3\}$ into a \ptt-free graph by deleting the at most~$k-2$ edges in~$S \setminus \{e_2,e_3\}$.
\end{proof}

\begin{branch} \label{BR: eliminate these crazy diamonds}
If there are vertices~$u,v,w,z \in V$ such that~$G[\{u,v,w,z\}]$ is an LC-Diamond (Figure~\ref{Figure: Branching Subgraphs} (a)) or an LO-Diamond (Figure~\ref{Figure: Branching Subgraphs} (b)) or an IIZ-Diamond (Figure~\ref{Figure: Branching Subgraphs} (c)), then branch into the cases
\begin{enumerate}
\item[•] $I_1:= (G-\{v,w\},k-1)$,
\item[•] $I_2:= (G-\{\{u,v\},\{u,z\}\},k-2)$, and
\item[•] $I_3:= (G-\{\{u,v\},\{v,z\}, \{w,z\}\},k-3)$.
\end{enumerate} 
\end{branch}

\begin{lemma}
Branching Rule~\ref{BR: eliminate these crazy diamonds} is correct.
\end{lemma}

\begin{proof}
We show that~$(G,k)$ is a yes-instance of \BPDs~if and only if at least one of the instances~$I_1$, $I_2$, or~$I_3$ is a yes-instance of \BPDs.

$(\Leftarrow)$ This direction holds since in every instance~$I_i$ the parameter~$k$ is decreased by the exact amount of edges deleted from~$G$.

$(\Rightarrow)$ Let~$S$ be a solution for~$G$. In LO-Diamonds, LC-Diamonds, and in IIZ-Dia\-monds, the edges~$\{u,v\}$ and~$\{v,w\}$ form a \pt~in~$G$ and therefore~$\{u,v\} \in S$ or~$\{v,w\} \in S$. If~$\{v,w\} \in S$, then~$I_1$ is a yes-instance. Otherwise, if~$\{v,w\} \not \in S$ it follows that~$\{u,v\} \in S$. If~$\{u,z\} \in S$, we have~$\{u,v\},\{u,z\} \in S$ and therefore~$I_2$ is a yes-instance. So we assume~$\{v,w\} \not \in S$, $\{u,v\} \in S$, $\{u,z\} \not \in S$ and consider the following cases.

\textbf{Case 1:} $G[\{u,v,w,z\}]$ is an LC-Diamond. Then, $\{u,z\}$ and~$\{v,z\}$ form a \pt~in~$G-\{u,v\}$. Since~$\{u,z\} \not \in S$, it follows that~$\{v,z\} \in S$. The edges~$\{v,w\}$ and~$\{w,z\}$ form a \pt~in~$G-\{\{u,v\},\{v,z\}\}$ which implies~$\{v,w\} \in S$ or~$\{w,z\} \in S$. Since~$\{v,w\} \not \in S$, we have~$\{w,z\} \in S$. Thus,~$\{u,v\},\{v,z\},\{w,z\} \in S$ and therefore~$I_3$ is a yes-instance.

\textbf{Case 2:} $G[\{u,v,w,z\}]$ is an LO-Diamond. Then, $\{u,z\}$ and~$\{w,z\}$ form a \pt~in~$G-\{u,v\}$. Since~$\{u,z\} \not \in S$, it follows that~$\{w,z\} \in S$. The edges~$\{v,w\}$ and~$\{v,z\}$ form a \pt~in~$G-\{\{u,v\},\{w,z\}\}$ which implies~$\{v,w\} \in S$ or~$\{v,z\} \in S$. Since~$\{v,w\} \not \in S$, we have~$\{v,z\} \in S$. Thus,~$\{u,v\},\{v,z\},\{w,z\} \in S$ and therefore~$I_3$ is a yes-instance.

\textbf{Case 3:} $G[\{u,v,w,z\}]$ is an IIZ-Diamond. Then, $\{u,z\}$ forms a \pt~with~$\{w,z\}$ and with~$\{v,z\}$ in~$G-\{u,v\}$. Since~$\{u,z\} \not \in S$, it follows that~$\{w,z\}, \{v,z\} \in S$ and therefore~$I_3$ is a yes-instance.
\end{proof}

\begin{branch} \label{BR: eliminate CC-Hourglasses}
If there are vertices~$u,v,w,z_1,z_2 \in V$ such that~$G[\{u,v,w,z_1,z_2\}]$ is a CC-Hourglass as given in Figure~\ref{Figure: Branching Subgraphs} (d), then branch into the cases
\begin{enumerate}
\item[•] $I_1:= (G-\{v,w\},k-1)$,
\item[•] $I_2:= (G-\{\{u,v\},\{v,z_1\}\},k-2)$, and
\item[•] $I_3:= (G-\{\{u,v\},\{u,z_1\}, \{v,z_2\}\},k-3)$.
\end{enumerate} 
\end{branch}

\begin{lemma}
Branching Rule~\ref{BR: eliminate CC-Hourglasses} is correct.
\end{lemma}

\begin{proof}
We show that~$(G,k)$ is a yes-instance of \BPDs~if and only if at least one of the instances~$I_1$, $I_2$, or~$I_3$ is a yes-instance of \BPDs.

$(\Leftarrow)$ This direction obviously holds since in every instance~$I_i$ the parameter~$k$ is decreased by the exact amount of edges deleted from~$G$.

$(\Rightarrow)$ Let~$S$ be a solution for~$G$. The edges~$\{u,v\}$ and~$\{v,w\}$ form a \pt~in~$G$ and~therefore~$\{u,v\} \in S$ or $\{v,w\} \in S$. If~$\{v,w\} \in S$, then~$I_1$ is a yes-instance. Otherwise, if~$\{v,w\} \not \in S$, then~$\{u,v\} \in S$. The edge~$\{v,z_1\}$ forms a \pt~with~$\{u,z_1\}$ and~$\{v,z_2\}$~in~$G-\{u,v\}$. If~$\{v,z_1\} \in S$, then~$I_2$ is a yes-instance. Otherwise, if~$\{v,z_1\} \not \in S$, then~$\{u,z_1\},\{v,z_2\} \in S$. Hence, $I_3$ is a yes-instance.
\end{proof}

We use the Branching Rules \ref{BR: eliminate multi pts}--\ref{BR: eliminate CC-Hourglasses} to state the following theorem.

\begin{theorem} \label{Theorem: Branching Algorithm}
  \BPDs{} can be solved in~$\Oh(1.84^k\cdot nm)$~time.
\end{theorem}
\begin{proof}
We solve \BPDs~for an instance~$(G,k)$ as follows: Initially, we compute the adjacency matrix of~$G$ in~$\Oh(n^2)$ time. We then compute one of the structures we branch on, which is an induced LC-Diamond, LO-Diamond, IIZ-Diamond, CC-Hourglass or some edge which forms a \pt~with two other edges. Next, we branch according to the Branching Rules~\ref{BR: eliminate multi pts}, \ref{BR: eliminate these crazy diamonds}, and \ref{BR: eliminate CC-Hourglasses}. If no further branching rule is applicable we check in~$\Oh(nm)$ time whether the remaining instance is a yes-instance or not. This is possible by Corollary~\ref{Corollary: Solve after Branching}. The branching vectors are~$(1,2)$ for Branching Rule~\ref{BR: eliminate multi pts}, and~$(1,2,3)$ for Branching Rules~\ref{BR: eliminate these crazy diamonds} and~\ref{BR: eliminate CC-Hourglasses}. This delivers a branching factor smaller than~$1.8393$. We next describe in detail how we can find one of the structures we branch on, in such a way that the algorithm described above runs in~$\mathcal{O}(1.84^k \cdot nm )$~time

Isolated vertices do not contribute to the solution of the instance. Thus, we delete all isolated vertices in~$\mathcal{O}(n)$ time. Hence we can assume~$n \leq 2m$ in the following. Afterwards, we compute a maximal packing~$\mathcal{P}$ of edge-disjoint~\pt s. Here, we represent a~\pt{} by an edge set of size two. We define~$P := \bigcup_{p \in \mathcal{P}} p$ as the set of all edges of~\pt s in~$\mathcal{P}$. Note that~$\mathcal{P}$ and thus~$P$ can be found in~$\Oh(n m)$ time by enumerating all~\pt s in~$G$. If~$|P| > 2k$, the graph~$G$ contains more than~$k$ edge-disjoint \pt s and~$(G,k)$ is a no instance. In this case we can stop and return no. Otherwise, we have~$|P| \leq 2k$ and use~$P$ to compute the structures we apply Branching Rules~\ref{BR: eliminate multi pts}--\ref{BR: eliminate CC-Hourglasses} on as follows.

Note that in any LC-, LO- or IIZ- Diamond~$G[\{u,v,w,z\}]$ the edges~$\{u,v\}$ and~$\{v,w\}$ form a \pt{} in~$G$ and thus, $\{u,v\} \in P$ or~$\{v,w\} \in P$. Therefore, we can find~$G[\{u,v,w,z\}]$ in $\Oh(km)$ time by iterating over all possible tuples~$(e_1,e_2)$ with~$e_1 \in P$ and $e_2 \in E$ and then check with the adjacency matrix in~$\Oh(1)$ time if the induced subgraph~$G[e_1 \cup e_2]$ is an LC-, LO- or IIZ- Diamond.

Furthermore, observe that any CC-Hourglass~$G[\{u,v,w,z_1,z_2\}]$ contains the two edge-disjoint \pt s~$\{\{u,v\}, \{v,w\}\}$ and~$\{\{z_1,v\},\{v,z_2\}\}$. Thus, at least two edges of~$G[\{u,v,w,z_1,z_2\}]$ are elements of~$P$. Therefore, we can find~$G[\{u,v,w,z_1,z_2\}]$ in~$\Oh(k^2n^2)$ time by iterating over all possible tuples~$(e_1,e_2,x,y)$ with~$e_1, e_2 \in P$ and~$x,y \in V$ and then check with the adjacency matrix in~$\Oh(1)$ time if the induced subgraph~$G[e_1 \cup e_2 \cup \{x,y\}]$ is a CC-Hourglass.

Finally, let~$e_1$ be an edge that forms two~\pt s with other edges~$e_2$ and~$e_3$. Again, at least one of the edges~$e_1, e_2$, or~$e_3$ is an element of~$P$. Thus, we can find~$e_1$, $e_2$, and~$e_3$ in~$\Oh(kn^2)$ time iterating over all triples containing an edge from~$P$ and two vertices from~$V$, which form the remaining two endpoints.

Since $n \leq 2m$ we can find one of the structures on which we apply Branching Rules~\ref{BR: eliminate multi pts}--\ref{BR: eliminate CC-Hourglasses} in~$\Oh(k^2 nm)$ time. This gives us a total running time of~$\Oh(1.8393^k \cdot k^2 nm) \subseteq \Oh(1.84^k \cdot nm)$ time as claimed.
\end{proof}

It is possible to improve the branching rules on LO-Diamonds, IIZ-Diamonds, and CC-Hourglasses to obtain a branching vector~$(2,2,3,3)$, but  branching on LC-Diamonds still needs a branching vector of~$(1,2,3)$, which is the bottleneck. To put the running time of Theorem~\ref{Theorem: Branching Algorithm} into perspective note that \textsc{Cluster Deletion}, which can be viewed as the uncolored version of \BPDs, can be solved in~$\Oh(1.42^k+m)$ time \cite{BD11}. Thus there is a large gap between the running time bounds of the problems. It would be interesting to know if this gap can be closed or if \BPDs~is significantly harder than \textsc{Cluster Deletion}.

\subsection{On Problem Kernelization} Finally, we consider problem kernelization for \BPDs~parameterized by~$(k,\Delta)$ and~$\ell:=m-k$. Recall that~$\Delta$ denotes the maximum degree of the input graph. We show that \BPDs~admits problem kernels with~$\Oh(k\Delta\min(k,\Delta))$~vertices or at most~$2\ell$ edges, respectively.

In the following, we provide two reduction rules leading to an~$\Oh(k\Delta\min(k,\Delta))$ vertex kernel for \BPDs. The first reduction rule deletes all edges which form more than~$k$ \pt s.

\begin{reduc}
\label{reduc:edge-in-too-many-conflicts}
If~$G$ contains an edge~$\{ u,v\}$ such that there exist vertices~$w_1,\ldots ,w_t$ with~$t >k$ such that~$G[\{ u,v,w_i\}]$ is a \pt~for each~$i$, then remove~$\{ u,v\}$ and decrease~$k$ by one.
\end{reduc}
\begin{lemma}
Reduction Rule~\ref{reduc:edge-in-too-many-conflicts} is correct and can be applied exhaustively in~$\Oh(nm)$~time.
\end{lemma}
\begin{proof}
First, we prove the correctness of Reduction Rule~\ref{reduc:edge-in-too-many-conflicts}. 
Let~$S$ be a solution for~$(G,k)$. 
Without loss of generality, consider an edge~$\{ u,v\}\in E_r$ such that there exist vertices~$w_1, \ldots , w_t$ such that for each~$i$ the graph~$G[\{ u,v,w_i\}]$ is a \pt. 
At least one edge of each \pt,~$G[\{ u,v,w_i\}]$ is an element of~$S$. 
Assume towards a contradiction~$\{ u,v\}\notin S$. 
In each \pt,~$G[\{ u,v,w_i\}]$ the blue edge which is either~$\{ v, w_i\}$ or~$\{u, w_i\}$ has to be removed. 
Note that for each~$w_i$ these are pairwise different edges.
Thus, since~$t >k$,~$|S|>k$, a contradiction to~$|S|\le k$. 
Hence,~$\{ u,v\}\in S$ and~$S\setminus\{\{ u,v\}\}$ is a solution for~$(G-\{ u,v\}, k-1)$. For the opposite direction, if~$S$ is a solution for~$(G-\{ u,v\}, k-1)$ then~$S\cup\{\{ u,v\}\}$ is is a solution for~$(G,k)$.

Second, we bound the running time of applying Reduction Rule~\ref{reduc:edge-in-too-many-conflicts} exhaustively. In a first step, for each edge~$e\in E$ compute the number of \pt s containing~$e$. This can be done in~$\Oh(nm)$~time. In a second step, check if an edge~$e=\{ u,v\}$ is part of more than~$k$ \pt s and remove~$e$ if this is the case. After the removal of~$e$, every new \pt~contains vertices~$u$ and~$v$. Hence, for each remaining vertex~$w\in V$, check if~$G[\{u,v,w\}]$ is a new \pt~in~$G-\{ u,v\}$. If yes, then update the number of \pt s for edges~$\{ u,w\}$ and~$\{ v,w\}$. This can be done in~$\Oh(n)$ time. 
Since~$k<m$, the overall running time of Reduction Rule~\ref{reduc:edge-in-too-many-conflicts} is~$\Oh(nm)$.
\end{proof}

Let~$\mathcal{P}$ denote the set of all vertices of~$G$ which are part of \pt s. Then, the set~$N[\mathcal{P}]$ contains all vertices which are either part of a \pt~or which are adjacent to a vertex in a \pt. In other words, a vertex~$v$ is contained in~$V\setminus N[\mathcal{P}]$ if and only if each vertex~$u\in N[v]$ is not part of any \pt . In the following, we present a reduction rule to remove all vertices in~$V\setminus N[\mathcal{P}]$.

\begin{reduc}
\label{reduc:edge-deletions-no-second-neighborhood}
If~$G$ contains a vertex~$v\in V$ such that each vertex~$u\in N[v]$ is not part of any \pt, then, remove~$v$ from~$G$.
\end{reduc}

To show that Rule~\ref{reduc:edge-deletions-no-second-neighborhood} is correct, we provide two simple lemmas about edge deletion sets.

\begin{lemma} \label{Lemma: Incident Deletion}
Let~$G=(V,E)$ be a graph and let~$S \subseteq E$ be an edge deletion set. If two edges~$\{u,v\} \in E$ and~$\{v,w\} \in E$ do not form a~\pt~in~$G$ and form a~\pt~in~$G-S$, then~$\{u,w\} \in S$.
\end{lemma}

The proof of Lemma~\ref{Lemma: Incident Deletion} is trivial and thus omitted.

\begin{lemma} \label{Lemma: deletion-sequence}
Let~$(G,k)$ be an instance of~\BPDs{} and let~$S$ be a solution for~$(G,k)$ of minimum size. Then, there exists an ordering~$(e_1, \dots, e_{|S|})$ of~$S$ such that for each~$i \in \{1, \dots, |S|\}$ the edge~$e_i$ is part of a~\pt~in~$G-\{e_1,\dots, e_{i-1}\}$.

\end{lemma}
\begin{proof}
Assume towards a contradiction that such an ordering does not exist. Then, for every ordering of~$S$, there exists a maximal index~$1\le i<|S|$ such that there is a finite sequence~$(e_1, \ldots , e_i)$ where for each~$1 \leq j \leq i$ the edge~$e_j$ is part of a \pt~in~$G-\{ e_1, \ldots , e_{j-1}\}$. According to our choice of~$i$, there exists no edge of~$S$ which is part of a \pt~in~$G-\{ e_1, \ldots , e_i\}$. If~$G-\{ e_1, \ldots , e_i\}$ contains a \pt{} formed by~$\{x,y\}$ and~$\{y,z\}$, then~$\{ x,y\} \not \in S$ and~$\{y,z\} \not \in S$. This contradicts the fact that~$S$ is a solution for~$(G,k)$. Otherwise,~$G-\{ e_1, \ldots , e_i\}$ is \ptt-free. Then,~$\{e_1, \ldots , e_i\} \subsetneq S$ is a solution for~$(G,k)$. This contradicts the fact that~$S$ is a solution of minimum size for~$(G,k)$.
\end{proof}

We next use the Lemmas~\ref{Lemma: Incident Deletion} and~\ref{Lemma: deletion-sequence} to show the correctness of Rule~\ref{reduc:edge-deletions-no-second-neighborhood}.

\begin{lemma}
Reduction Rule~\ref{reduc:edge-deletions-no-second-neighborhood} is correct and can be applied exhaustively in~$\Oh(nm)$~time.
\end{lemma}
\begin{proof}

Let~$H:=G[V\setminus\{ v\}]$. We prove that there exists a solution~$S$ for~$(G,k)$ if and only if~$S$ is also a solution for~$(H,k)$. 

$(\Rightarrow )$ Let~$S$ be a solution for~$(G,k)$. Since~$H$ is an induced subgraph of~$G$,~$S\cap E(H)$ is a solution for~$(H,k)$. 

$(\Leftarrow )$ Let~$S$ be a solution for~$(H,k)$ of minimum size. If~$S=\emptyset$,~$S$ is also a solution for~$(G,k)$. Thus, in the following we assume that~$S\neq\emptyset$. We show that~$G-S$ is~\ptt-free. To this end, we provide a claim to show that there are no edge-deletions in the neighborhood of~$v$.

\begin{claim} \label{Claim: Deletions in S not in N[v]}
In~$S$, there are no edge deletions that are incident with some vertex in~$N_G[v]$.
\end{claim}

\begin{claimproof}
Assume towards a contradiction that~$S$ contains such edge deletions. We consider an ordering~$F=(e_1, \dots, e_{|S|})$ from Lemma~\ref{Lemma: deletion-sequence} such that every~$e_i$ is part of a~\pt~in~$H-\{e_1, \dots, e_{i-1}\}$. Note that this implies that every~$e_i$ is in a~\pt~in~$G-\{e_1, \dots, e_{i-1}\}$. Let~$t \in \{1, \dots, |S|\}$ be the minimum index such that~$e_t=:\{w,z\}$ is incident with some~$w \in N_G[v]$.  

Consider the case~$z \in N_G[v]$. Then, since~$\{w,z\}$ is part of a~\pt~in~$G-\{e_1, \dots, e_{t-1}\}$ and~$\{w,z\}$ is not part of a~\pt~in~$G$, Lemma~\ref{Lemma: Incident Deletion} implies that there exists some index~$j<t$ such that~$e_j \in S$ is incident with~$w$ or with~$z$. This contradicts the minimality of~$t$. Thus, we may assume~$z \in N_G^2(v)$ and therefore~$w \in N_G(v)$ for the rest of this proof. Without loss of generality we may assume that~$\{v,w\}$ is red. Observe that this implies that~$\{w,z\}$ is red, since~$v$ is not part of any~\pt~in~$G$.

Since~$\{w,z\}$ is part of a~\pt~in~$G-\{e_1, \dots, e_{t-1}\}$, there exists some vertex~$y$ such that~$G-\{e_1, \dots, e_{t-1}\}[w,z,y]$ is a~\pt. We show that the following cases are contradictory.

\textbf{Case 1:}~$y \in N_G(v)$\textbf{.}  Then, since~$G-\{e_1, \dots, e_{t-1}\}[w,z,y]$ is a \pt and~$G[w,z,y]$ is not a \pt, Lemma~\ref{Lemma: Incident Deletion} implies that there exists some index~$j<t$ and there is an edge~$e_j$ in the ordering~$F$ that is incident with~$y$ or with~$w$. This contradicts the minimality of~$t$ since~$y$ and~$w$ are elements of~$N_G(v)$.

\textbf{Case 2:}~$y \in N^2_G(v)$\textbf{.} Then,~$\{w,z\}$ and~$\{w,y\}$ do not form a~\pt~in~$G-\{e_1, \dots, e_{t-1}\}$ since all edges in~$E_G(\{w\},N_G^2(v))$ are red since~$v$ is not part of any~\pt~in~$G$. Consequently,~$\{w,z\}$ and~$\{z,y\}$ form a~\pt~in~$G-\{e_1, \dots, e_{t-1}\}$. Then, analogous to Case~1, Lemma~\ref{Lemma: Incident Deletion} implies that there exists some~$e_j$ with~$j<t$ and~$e_j=\{w,y\}$ which is a contradiction to the minimality of~$t$.

\textbf{Case 3:}~$y \in V \setminus (N_G(v) \cup N^2_G(v))$\textbf{.} Then,~$\{w,z\}$ and~$\{z,y\}$ form a~\pt~in~$G-\{e_1, \dots, e_{t-1}\}$. Moreover,~$\{w,y\} \not \in E$, since~$y \not \in N_G^2(v)$. Then,~$\{w,z\}$ and~$\{z,y\}$ also form a~\pt~in~$G$ contradicting the fact that no vertex in~$N_G[v]$ is part of a~\pt. $\hfill \Diamond$
\end{claimproof}

We next use Claim~\ref{Claim: Deletions in S not in N[v]} to show that~$G-S$ is~\ptt-free. Observe that it suffices to show that no edge incident with~$v$ is part of a~\pt~in~$G-S$. Consider an edge~$\{u,v\}\in E$. Then, by Claim~\ref{Claim: Deletions in S not in N[v]},~$u$ and~$v$ are incident with the same edges in~$G$ as in~$G-S$. Therefore, since~$\{u,v\}$ is not part of any~\pt~in~$G$, the edge~$\{u,v\}$ is not part of any~\pt~in~$G-S$. Consequently,~$S$ is a solution for~$(G,k)$.

It remains to consider the running time of applying Reduction Rule~\ref{reduc:edge-deletions-no-second-neighborhood} exhaustively. In a first step, determine all \pt s in~$G$. Afterwards, determine for each vertex~$v\in V$ if~$v$ is part of some \pt. This needs~$\Oh(nm)$~time. Now, check for each vertex~$v\in V$ if each vertex~$u\in N[v]$ is not part of any \pt. This can be done in~$\Oh(m)$~time. The claimed running time follows.
\end{proof}

\begin{theorem}
\label{theo:k2delta-kernel-bip3del}
\BPDs~admits a~$\Oh(k\Delta\min(k,\Delta))$-vertex kernel that can be computed in~$\mathcal{O}(nm)$~time.
\end{theorem}
\begin{proof}
First, apply Reduction Rule~\ref{reduc:edge-in-too-many-conflicts} exhaustively. Second, apply Reduction Rule~\ref{reduc:edge-deletions-no-second-neighborhood} exhaustively. This needs~$\Oh(nm)$~time altogether. We prove that~$G$ contains at most $12 \cdot k\Delta\min(k,\Delta)$~vertices if~$(G,k)$ is a yes-instance. Let~$\mathcal{P}$ be the set of vertices which are contained in a \pt~in~$G$, and let~$P$ be a maximal set of edge-disjoint \pt s in~$G$. If~$(G,k)$ is a yes-instance for \BPDs, then~$|P|\le k$. Since Reduction Rule~\ref{reduc:edge-in-too-many-conflicts} was applied exhaustively, each edge is part of at most~$\min(k,2\Delta)$ \pt s. Hence, the total number of \pt s in~$G$ is at most $2k\min(k,2\Delta)$. Consequently,~$|\mathcal{P}|\le 6k\min(k,2\Delta)$. Since Reduction Rule~\ref{reduc:edge-deletions-no-second-neighborhood} was applied exhaustively,~$V\setminus N[\mathcal{P}]=\emptyset$. In other words, set~$\mathcal{P}$ has no second neighborhood in~$G$. Since each vertex has degree at most~$\Delta$ we have $|{N}(\mathcal{P})|\le 6k\Delta\min(k,2\Delta)$. Hence, the overall number of vertices in~$G$ is~$12 \cdot k\Delta\min(k,\Delta)$ if~$(G,k)$ is a yes-instance for \BPDs.
\end{proof}
By the above,  \BPDs~admits a linear problem kernel in~$k$ if~$G$ has constant maximum degree.
Note that a kernelization by~$\Delta$ alone is unlikely since \BPDs~is NP-hard even if~$\Delta=8$ by Theorem~\ref{Theorem: NP-h}.  Since \BPDs~is fixed-parameter tractable with respect to parameter~$k$, we can trivially conclude that it admits an exponential-size problem kernel. It is open if there is a polynomial kernel depending only on~$k$ while \textsc{Cluster Deletion} has a relatively simple~$4k$-vertex kernel \cite{Guo09}. Summarizing, \BPDs~seems to be somewhat harder than \textsc{Cluster Deletion} if parameterized by~$k$. 

In contrast, \BPDs~seems to be easier than \textsc{Cluster Deletion} if parameterized by the dual parameter~$\ell:=m-k$: there is little hope that \textsc{Cluster Deletion} admits a problem kernel of size~$\ell^{\Oh(1)}$~\cite{GK18} while \BPDs{} has a trivial linear-size problem
kernel as we show below.

\begin{theorem}
  \BPDs{} admits a problem kernel with~$2\ell$ edges and~vertices which can be computed in~$\Oh(n+m)$ time.
\end{theorem}
\begin{proof}
We show that instances with at least~$2\ell$ edges are trivial yes-instances. Let~$(G=(V,E_r,E_b),k)$ with~$|E| \geq 2\ell$ be an instance of \BPDs. Then, since~$E_r$ and~$E_b$ form a partition of~$E$, we have~$|E_r| \geq \ell$ or~$|E_b| \geq \ell$. Without loss of generality let~$|E_r| \geq \ell$. Since~$|E_b| = m - |E_r| \leq m - \ell = k$,~$E_b$ is a solution for~$(G,k)$.
\end{proof}


\section{Outlook}
We have initiated the algorithmic study of a natural edge-deletion problem on edge-colored graphs. In companion work, we considered the problem of destroying paths of length at least 4 that fulfill certain coloring constraints~\cite{EGKS21,Eck20}. With this exception, however, the study of graph modification problems on edge-colored graphs has been neglected so far. Consequently, the complexity of many natural problems and a study of natural edge-colored graph classes remain open. 

For the particular case of~\ptt-free graphs, we have also left open many questions for future work. First, it would be interesting to further
investigate the structure of \ptt-free graphs. Since each color class may
induce an arbitrary graph it seems difficult to obtain a concise and non-trivial structural
characterization of these graphs. One may, however, exploit the connection with Gallai
colorings which are colorings where no triangle receives more than two colors. In particular, the following characterization of Gallai colorings is known~\cite{Gal67,GS04}: any Gallai
coloring of a complete graph can be obtained by taking some complete 2-colored graph and substituting its vertices with a complete graph and some Gallai
coloring of this complete graph. This characterization relies on the
decomposition property that in any Gallai coloring of a complete graph~$G$ with at least three
colors, there is at least one edge color that spans a disconnected graph~$H$. Then, by the
property of Gallai colorings, the edges in~$G$ that are between two different components~$H_1$
and~$H_2$ of~$H$ have the same color. 

Second, there are many open questions concerning \textsc{\Pt{} Deletion}. Does \textsc{\Pt{} Deletion} admit a polynomial-size kernel for~$k$? Can \textsc{\Pt{}
  Deletion} be solved in~$2^{\Oh(n)}$ time? Can \textsc{\Pt{} Deletion} be solved in
polynomial time on graphs that contain no monochromatic~$P_3$? Can \textsc{\Pt{}
  Deletion} be solved in polynomial time on graphs that contain no cycle consisting only
of blue edges? Even simpler is the following question: Can \textsc{\Pt{} Deletion} be
solved in polynomial time if the subgraphs induced by the red edges and the subgraph induced by the blue edges are each a disjoint union of paths?
Moreover, it would be interesting to perform a similar study on \textsc{\Pt{} Editing} where we
may also insert blue and red edges. Furthermore, it is open if \textsc{Bicolored-$P_3$-free Completion} where we only may insert red or blue edges is NP-hard. Observe in this context that the uncolored problem \textsc{Cluster Completion} can easily be solved by adding all missing edges in each connected component.

Third, it would also be interesting to identify graph problems that are NP-hard on general
two-edge colored graphs but polynomial-time solvable on \ptt-free graphs. 
Finally, we were not able to resolve the following question: Can we find \pt{}s in~$\Oh(n+m)$ time?

Using the connection to Gallai colorings and the decomposition property of Gallai colorings seems to be a
promising approach to address these open questions.
\medskip


\end{document}